\documentclass[11pt]{article}
\usepackage[sort&compress, authoryear]{natbib}
\usepackage{subfig}
\usepackage[utf8]{inputenc}
 
\usepackage[english]{babel}%
\usepackage{epsfig}%
\usepackage{url}%
\usepackage{amsmath}%
\usepackage{amsfonts}%
\usepackage{mathrsfs}%
\usepackage{amssymb}%
\usepackage{longtable}%
\usepackage{geometry, calc, color, setspace}%
\usepackage{indentfirst}%
\usepackage{bm}
\usepackage{multirow}
\usepackage{amsthm}
\usepackage{amsmath}
\usepackage{paralist}
\usepackage{placeins}

\theoremstyle{definition}

\newtheorem{theorem}{Theorem}[section]

\newtheorem{lem}{Lemma}[section]

\newtheorem{proposition}{Proposition}[section]

\newcommand{\beqn}{\begin{eqnarray*}}
\newcommand{\eeqn}{\end{eqnarray*}}
\newcommand{\beqnn}{\begin{eqnarray}}
\newcommand{\eeqnn}{\end{eqnarray}}

\begin{document}

\title {{\bf Some homogeneity test statistics for DNA evolutionary models}}
\author{Bordin, Tatiana B.$^{*}$;   
Pinheiro, Hildete P.$^{**}$ and 
Pinheiro, Alu\' {\i}sio$^{**}$\\
$^{*}$Dell Technologies, $^{**}$University of Campinas}

\maketitle

\begin{abstract}  
{

We present a test statistic for the comparison of DNA sequences under some of the most popular evolutionary processes available in the literature. Theoretical properties for the test statistic as well as its empirical performance by stochastic simulations are presented. 

The proposed test statistic is a generalized $U$-statistics built for tests of distributional homogeneity under the null hypothesis. We show that a dicothomous situation exists here. Under the null hypothesis, the $U$-statistics kernel is first-order degenerated, this test statistic falls in the quasi $U$-statistics class and follows an asymptotic normal law, albeit of higher order than the standard case. Under heterogeneity, the asymptotic normality is attained on the more usual first-order asymptotics. 

Asymptotic normality is proven for the cases: high-dimension/large sample size,  high-dimension/small sample size, low-dimension/large sample size. Moreover, the case of local alternatives is discussed, and the contiguity of the test statistic for them is established. Simulation studies are performed to assess some finite-dimensional properties of the test statistic, regarding issues such as balanced/unbalanced samples, dimension and sample size. 
}
\medskip

\noindent{\textit{Keywords}}: {Diversity Measures, Statistical genetics, 
Quasi U-statistics, Mixing condition, Exchangeable variables, substitution models.}
\end{abstract}

\section{Introduction}\label{intro}

One of the goals of genetic studies is to compare groups through genetic variability. In many studies in the literature (\citealp{Tavare:2004}; \citealp{Excoffier:1992}), genetic variability is found by estimating the expected number of mutations per site and per generation, say $\theta$. Some models were proposed to estimate this parameter, such as the Wright-Fisher model (\citealp{Wright:1949, Ewens:1972}), K-alleles model, 
infinite alleles model (\citealp{Kimura:1969}), infinite sites model (\citealp{Kimura:1971}) among others. 
Under the infinite sites model, the parameter $\theta$ can be interpreted as the expected number of mutation per site per generation.

Some problems arise when we are working with DNA sequences. First, the sites can be related by their molecular use and, deterministic and stochastic associations should be expected. Another source of complexity added to the analysis is the fact that one may have a sample of $n$ sequences with lengths $K$, where $K >> n$.

The main contribution of this work is to evaluate the asymptotic distribution of a test statistic based on measures of divergence between DNA sequences. In the context of homogeneity tests among groups, proposed by \citet{Pinheiro:2005, Pinheiro:2009, Pinheiro:2011}, we will incorporate the evolutionary processes of the sequences. 


Particular interest also lies in the properties and distribution of the test statistic proposed in \citet{Pinheiro:2009} considering DNA sequences dependency. Asymptotic theory based on independent and identically distributed (i.i.d.) sequences should not be applied in this context. It is necessary the assumption of interchangeable variables (\citealp{DeFinetti:1937}). 

The manuscript is organized as follows: Section \ref{diversity} presents the biological motivation, 
the decomposition of genetic diversity, the test statistic for the homogeneity test of groups of DNA sequences and their moments under an evolutionary process. 
The asymptotic distribution of the test statistic under the null and alternative hypothesis is proven in Section{teoAss}. Section \ref{Simul} presents the results of a simulation study to verify the behavior of the test statistic under the null and alternative hypothesis for finite sample sizes.  A discussion of the asymptotic normality of the test statistic is also considered in Section \ref{sec5}.

\section{Genetic Diversity and Decompositions} \label{diversity}

\citet{Pinheiro:2005, Pinheiro:2009, Pinheiro:2011} considers the comparison of $G$ ($\geq 2$) groups of DNA sequences. It is there assumed that sequences are independent, although there may be some structure of dependence within each sequence. We further develop this tests' theory by relaxing the independence assumption through the explicit development of some evolutionary DNA models.

Multiple mutations at the same site may be incorporated into genetic distances between sequences by substitution models such as Jukes-Cantor (\citealp{Jukes:1969}), Kimura (\citealp{Kimura:1980}), Felsenstein (\citealp{Felsenstein:1981}) or Hasegawa, Kino and Yano - HKY (\citealp{Hasegawa:1985}). These models try to describe mathematically the evolution process of DNA sequences over time. Two types of substitutions may occur: transition (when substitutions are between
purines or between pyrimidines, i.e., $\{A\leftrightarrow G, T\leftrightarrow C\}$) or transversion (when substitutions are between a
purine and a pyrimidine, i.e., $\{A,G\leftrightarrow T,C\}$). For instance, under the Jukes-Cantor model (JC69), the rates of transition and transversion are assumed to be equal and the  base frequencies are all equal ($\pi_A=\pi_C=\pi_T=\pi_G=1/4$); under the Kimura model (K80), the rates of transition and transversion are different and the base frequencies are equal; under Felsenstein model (F81), the rates of transition and transversion are equal but the base frequencies are arbitrary; and under the HKY model (HKY85) transition and transversion rates are different and the base frequencies are arbitrary. 
  
Suppose the groups are originated as in Figure \ref{figcoalseq2} (with no loss of generality drawn for 
 $G=5$ groups and one sequence for each subject at time $t$). For statistical purposes, let $\mathbf{X}_{gi}(t)=(X_{gi1}(t), \ldots, X_{giK}(t))$ be the random vector representing the $i$-th DNA sequence with $K$ sites for the $g$-th group at time $t$, $i=1, \ldots, n_g$. The $X_{gik}(t)$ is the random variable (r.v.) representing the nucleotide at site $k$, of the $i$-th sequence of group $g$ at time $t$.


\begin{figure}[ht]
\centering
\includegraphics[width=1.0\linewidth]{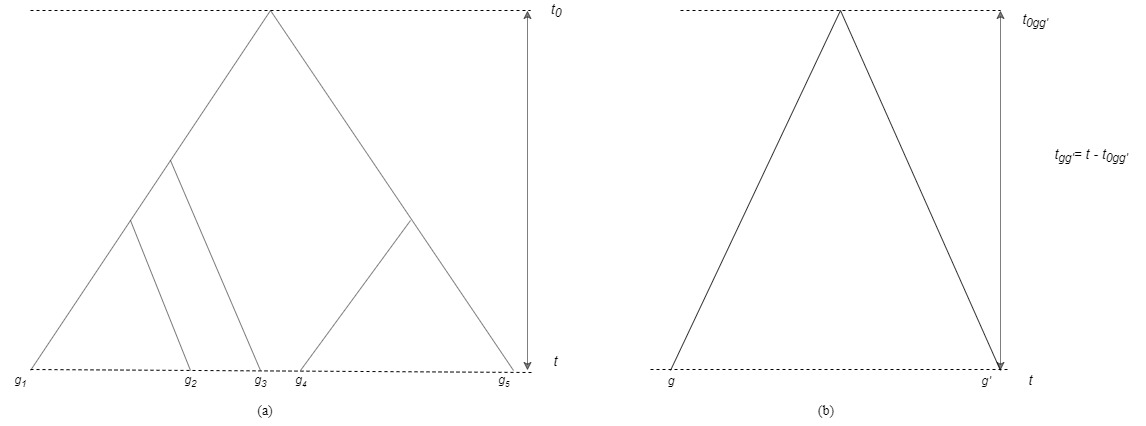}
\caption{(a) An example of evolutionary process for five different groups, (b) Phylogenetic tree of two Groups, with a common ancestor.}
\label{figcoalseq2}
\end{figure}

Consider two sequences observed at time $t$, one from group $g$ and the other from group $g'$, $\mathbf{X}_{gi} (t)$ and $\mathbf{X}_{g'j}(t)$, respectively. Suppose the sequences have an ancestor group at time $t_{ogg'}$, according to Figure \ref{figcoalseq2}(b). So, the empirical distance between $\mathbf{X}_{gi} (t)$ and $\mathbf{X}_{g'j}(t)$ is the Hamming distance given by
\beqn
D^{(gg')}_{ij}(t)=\frac{1}{K}\sum_{l=1}^K \mathbb{I}(X_{gil}(t)\neq X_{g'jl}(t)),
\eeqn
where $X_{gil}(t)$ ($X_ {g'jl} (t)$) represents the response at site $l$ of sequence $i$ ($j$) of group $g$ ($g'$) at time $t$. Denote by ``$o$'' the ancestral group between $g$ and $g'$ at $t_{ogg'}$ and let  $t_{gg'}=t-t_{ogg'}$ the time of MRCA (most recent common ancestor) between groups $g$ and $g'$. Denote by $X_{ok}(t_{ogg'})$, the response at site $l$ in the ancestral sequence ``$o$'' at time $t_{ogg'}$. The first moment of $D_{ij}^{(gg')}(t)$ is given by
\begin{eqnarray}
& &E(D^{(gg')}_{ij}(t)) = \frac{1}{K}\sum_{l=1}^K P(X_{gil}(t)\neq X_{g'jl}(t))=\frac{1}{K}\sum_{l=1}^K \sum_{c=1}^C P(X_{gil}(t)\neq X_{g'jl}(t),X_{ol}(t_{ogg'})=c)\nonumber\\
&=&\frac{1}{K}\sum_{l=1}^K\sum_{c=1}^C \left [ P(X_{ol}(t_{ogg'})=c)P(X_{gil}(t)\neq c \mid X_{ol}(t_{ogg'})=c) P(X_{g'jl}(t)=c\mid X_{ol}(t_{ogg'})=c) \right .\nonumber \\
&&+P(X_{ol}(t_{ogg'})=c) P(X_{gil}(t)= c \mid X_{ol}(t_{ogg'})=c)P(X_{g'jl}(t)\neq c \mid X_{ol}(t_{ogg'})=c)\nonumber \\
&&+P(X_{ol}(t_{ogg'})=c)\sum_{d\neq c}P(X_{gil}(t)= d \mid X_{ok}(t_{ogg'})=c)\sum_{\substack{e\neq d\\ e\neq c}}P(X_{g'jl}(t)= e\mid X_{ol}(t_{ogg'})=c)]. \nonumber 
\end{eqnarray}
Now, let
\begin{eqnarray}
\pi_{olc}&=&P(X_{ol}(t_{ogg'})=c),\label{eqpikc}\\
\pi_{g\mid o;ld\mid c}(t_{gg'})&=&P(X_{gil}(t)=d \mid X_{ol}(t_{ogg'})=c)\;\;\;\mbox{for}\;\;\;c,d = 1, \ldots, C. \label{eqcondition}
\end{eqnarray}
Therefore,
\begin{equation}
E(D^{(gg')}_{ij}(t))=1-\frac{1}{K}\sum_{l=1}^K\sum_{c=1}^C\pi_{olc}\sum_{d=1}^C\pi_{g \mid o;ld\mid c}(t_{gg'})\pi_{g'\mid o;ld \mid c}(t_{gg'}).
\end{equation}

The average sample distance between groups $g$ and $g'$ at time $t$ is
\begin{equation}
\bar{D}_{gg'}(t)=\frac{1}{n_g n_{g'}}\sum_{i,j}^{n_g,n_{g'}}D^{(gg')}_{ij}(t)\label{D_ggl_f}
\end{equation}
and its first moment is given by
\begin{equation}
E(\bar{D}_{gg'}(t))=1-\frac{1}{K}\sum_{l=1}^K\sum_{c=1}^C\pi_{olc}\sum_{c'=1}^C\pi_{g \mid o;lc'\mid c}(t_{gg})\pi_{g'\mid o;lc' \mid c}(t_{g'g'})=\mathcal{H}_{gg'}^\circ, \label{espDggl}
\end{equation}
where $\mathcal{H}_{gg'}^\circ$ is the population diversity between groups $g$ and $g'$ under the coalescence process, and $t_{gg}=\max\{t_{igg}, i=1, 2, \dots, n_{g}\}$, for group $g=1,\ldots, G$. Note that  $t_{igg}=t-t_{ogg}$, $i=1, \ldots, n_{g}$ represents the time of MRCA for the sequence on group $g$. 

Analogously, we have
\begin{eqnarray}
\bar{D}_{gg}(t)=\frac{1}{\binom{n_g}{2}}\sum_{i<j}^{1,n_{g}}\frac{1}{K}\sum_{l=1}^K \mathbb{I}(X_{gil}(t)\neq X_{gjl}(t)), \label{D_gg_f}
\end{eqnarray}
and
\begin{eqnarray}
E(\bar{D}_{gg}(t))&=& 1-\frac{1}{K}\sum_{l=1}^K\sum_{c=1}^C\pi_{olc}\sum_{c'=1}^C\pi^2_{g\mid o;lc' \mid c}(t_{gg})=\mathcal{H}^\circ_{gg}.  \label{espDgg}
\end{eqnarray}

Consider the comparison of $G \ge 2$ groups. If $X_{gil}(t)$ is the random variable representing the nucleotide at site $l$ of sequence $i$ in group $g$ at time $t$, $n_g$ is the total number of individuals in group $g$, and $n=\sum_{g=1}^G n_g$, the total number of sequences in the sample, then the pooled sample divergence is given by
\beqn
D_{n}(t,0)=\binom{n}{2}^{-1}\sum_{1\le i < j \le n} D_{ij}(t),
\eeqn
where $D_{ij}(t)=K^{-1}\sum_{l=1}^K\mathbb{I}(X_{il}(t)\neq X_{jl}(t))$ is the Hamming distance between two sequences of the individuals $i$ and $j$ at time $t$. Following \citet{Pinheiro:2005} the total divergence can be decomposed as follows
\beqn
D_n(t,0)&=& \sum_{g=1}^G\frac{n_g}{n} \bar{D}_{gg}(t) + \frac{1}{n(n-1)}\sum_{g=1}^{G-1}\sum_{g'=g+1}^G n_g n_{g'}(2\bar{D}_{gg'}(t)-\bar{D}_{gg}(t)-\bar{D}_{g'g'}(t))  \\
&=&D_n(t,W)+D_n(t,B),  
\eeqn
where $D_n(t,W)$ is the average sample within-groups Hamming distance, and $D_n(t,B)$ is the residual betwee-groups Hamming distance, written as: 
\beqnn
D_n(t,B)&=&\frac{1}{n(n-1)}\sum_{g=1}^{G-1}\sum_{g'=g+1}^G n_g n_{g'}(2\bar{D}_{gg'}(t)-\bar{D}_{gg}(t)-\bar{D}_{g'g'}(t))\;\;\;\mbox{and}\label{dntb_equation}\\
D_n(t,W)&=&\sum_{g=1}^G\frac{n_g}{n} \bar{D}_{gg}(t), \nonumber
\eeqnn
with $\bar{D}_{gg'}(t)$ and $\bar{D}_{gg}(t)$ given by (\ref{D_ggl_f}) and (\ref{D_gg_f}). 

Considering the coalescence process, (Figure \ref{figcoalseq2}(b)), the null hypothesis of homogeneity among groups is given by 
$H_0: 2 \mathcal{H}^\circ_{gg'}-\mathcal{H}^\circ_{gg}-\mathcal{H}^\circ_{g'g'}=0$, for 
$\forall g \neq g'$. 

\begin{proposition}
Under the assumption of the coalescence process, the alternative hypothesis of the homogeneity test is unilateral, i.e., $H_1: 2\mathcal{H}_{gg'}-\mathcal{H}_{gg}-\mathcal{H}_{g'g'}>0$.
\end{proposition}

\begin{proof}: 
Using the results from (\ref{espDggl}) and (\ref{espDgg}), it follows that
\beqn
\mathcal{H}^\circ_{gg'}&=&\frac{1}{K}\sum_{l=1}^K\left\{1-\sum_{c=1}^C\pi_{olc}\sum_{c'=1}^C\pi_{g\mid o;lc'\mid c}(t_{gg})\pi_{g'\mid o; lc'\mid c}(t_{g'g'})\right\}\\
&\geq&\frac{1}{K}\sum_{l=1}^K\left\{1-\frac{1}{2}\sum_{c=1}^C\pi_{olc}\sum_{c'=1}^C\left[\pi^2_{g\mid o;lc'\mid c}(t_{gg})+\pi^2_{g'\mid o;lc'\mid c}(t_{g'g'})\right]\right\}\\
&=&\frac{1}{2}\left\{1-\frac{1}{K}\sum_{l=1}^K\sum_{c=1}^C\pi_{olc}\sum_{c'=1}^C\pi^2_{g\mid o;lc'\mid c}(t_{gg})+1-\frac{1}{K}\sum_{l=1}^K\sum_{c=1}^C\pi_{olc}\sum_{c'=1}^C\pi^2_{g'\mid o;lc' \mid c}(t_{g'g'})\right\}\\
&=&\frac{1}{2}\left[\mathcal{H}^\circ_{gg}+\mathcal{H}^\circ_{g'g'}\right].
\eeqn
Note that $0\le(\pi_{g\mid o;lc'\mid c}(t_{gg})-\pi_{g'\mid o;lc'\mid c}(t_{g'g'}))^2$ 
$=-2\pi_{g\mid o;lc'\mid c}(t_{gg})\pi_{g'\mid o; lc'\mid c}(t_{g'g'})+\pi^2_{g\mid o;lc'\mid c}(t_{gg})+\pi^2_{g'\mid o;lc'\mid c}(t_{g'g'})$.
\end{proof}
Considering the HKY model (\citealp{Hasegawa:1985}), 
substitutions occur at rate $\alpha_l$ for transition, $\beta_l$ for transversion and $\pi_{lc}$ for $c\in \mathcal{S}=\{A,C,G,T\}$ and $l=1,\ldots,K$. Then, we have the following relation
\begin{itemize}
\item
Nucleotide $T$:
\beqnn
\pi_{lT\mid T}(t)&=&\pi_{lT}\left[1+\pi_{l\mathcal{R}}\pi_{l\mathcal{Y}}^{-1}\exp(-\beta_lt)\right]+\pi_{lC}\pi_{l\mathcal{Y}}^{-1}\exp(-\pi_{l\mathcal{R}}\beta_lt)\exp(-\pi_{l\mathcal{Y}}\alpha_lt) \nonumber \\
\pi_{lC\mid T}(t)&=&\pi_{lC}\left[1+\pi_{l\mathcal{R}}\pi_{l\mathcal{Y}}^{-1}\exp(-\beta_lt)-\pi_{l\mathcal{Y}}^{-1}\exp(-\pi_{l\mathcal{R}}\beta_lt)\exp(-\pi_{l\mathcal{Y}}\alpha_lt)\right] \nonumber \\
\pi_{lA\mid T}(t)&=&\pi_{lA}\left[1-\exp(-\beta_lt)\right] ~~~~~~~~~~~~ \pi_{lG\mid T}(t)=\pi_{lG}\left[1-\exp(-\beta_lt)\right] \label{NucTHKY}
\eeqnn
\item
Nucleotide $C$:
\beqnn
\pi_{lT\mid C}(t)&=&\pi_{lT}\left[1+\pi_{l\mathcal{R}}\pi_{l\mathcal{Y}}^{-1}\exp(-\beta_lt)-\pi_{l\mathcal{Y}}^{-1}\exp(-\pi_{l\mathcal{R}}\beta_lt)\exp(-\pi_{l\mathcal{Y}}\alpha_lt)\right] \nonumber \\
\pi_{lC\mid C}(t)&=&\pi_{lC}\left[1+\pi_{l\mathcal{R}}\pi_{l\mathcal{Y}}^{-1}\exp(-\beta_lt)\right]+\pi_{lT}\pi_{l\mathcal{Y}}^{-1}\exp(-\pi_{l\mathcal{R}}\beta_lt)\exp(-\pi_{l\mathcal{Y}}\alpha_lt) \nonumber \\
\pi_{lA\mid C}(t)&=&\pi_{lA}\left[1-\exp(-\beta_lt)\right] ~~~~~~~~~~~~\pi_{lG\mid C}(t)=\pi_{lG}\left[1-\exp(-\beta_lt)\right]\label{NucCHKY}
\eeqnn
\item
Nucleotide $A$:
\beqnn
\pi_{lT\mid A}(t)&=&\pi_{lT}\left[1-\exp(-\beta_lt)\right] ~~~~~~~~~~~~ \pi_{lC\mid A}(t)=\pi_{lC}\left[1-\exp(-\beta_lt)\right] \nonumber \\
\pi_{lA\mid A}(t)&=&\pi_{lA}\left[1+\pi_{l\mathcal{Y}}\pi_{l\mathcal{R}}^{-1}\exp(-\beta_lt)\right]+\pi_{lG}\pi_{l\mathcal{R}}^{-1}\exp(-\pi_{k\mathcal{Y}}\beta_lt)\exp(-\pi_{l\mathcal{R}}\alpha_lt) \nonumber \\
\pi_{lG\mid A}(t)&=&\pi_{lG}\left[1+\pi_{l\mathcal{Y}}\pi_{l\mathcal{R}}^{-1}\exp(-\beta_lt)-\pi_{l\mathcal{R}}^{-1}\exp(-\pi_{l\mathcal{Y}}\beta_lt)\exp(-\pi_{l\mathcal{R}}\alpha_lt)\right]\label{NucAHKY}
\eeqnn
\item
Nucleotide $G$:
\beqnn
\pi_{lT\mid G}(t)&=&\pi_{lT}\left[1-\exp(-\beta_lt)\right] ~~~~~~~~~~~~ \pi_{lC\mid G}(t)=\pi_{lC}\left[1-\exp(-\beta_lt)\right] \nonumber \\
\pi_{lA\mid G}(t)&=&\pi_{lA}\left[1+\pi_{l\mathcal{Y}}\pi_{l\mathcal{R}}^{-1}\exp(-\beta_lt)-\pi_{l\mathcal{R}}^{-1}\exp(-\pi_{l\mathcal{Y}}\beta_lt)\exp(-\pi_{l\mathcal{R}}\alpha_lt)\right] \nonumber \\
\pi_{lG\mid G}(t)&=&\pi_{lG}\left[1+\pi_{l\mathcal{Y}}\pi_{l\mathcal{R}}^{-1}\exp(-\beta_lt)\right]+\pi_{lA}\pi_{l\mathcal{R}}^{-1}\exp(-\pi_{l\mathcal{Y}}\beta_lt)\exp(-\pi_{l\mathcal{R}}\alpha_lt)\label{NucGHKY}
\eeqnn
\end{itemize}

Under HKY model, $\mathcal{H}^\circ_{gg}$ and $\mathcal{H}^\circ_{gg'}$ are given by
\beqn
&&\mathcal{H}^\circ_{gg}=1-\frac{1}{K}\sum_{k=1}^K\left\{\sum_{c=1}^C\pi_{okc}^2+\exp(-2\beta_{gk}t_{gg})\left[\frac{\pi_{ok\mathcal{R}}}{\pi_{ok\mathcal{Y}}}\sum_{c'\in\mathcal{Y}}\pi^2_{okc'}+\frac{\pi_{ok\mathcal{Y}}}{\pi_{ok\mathcal{R}}}\sum_{c'\in\mathcal{R}}\pi^2_{okc'}\right]\right.\\
&+&\left.\frac{2\pi_{okA}\pi_{okG}}{\pi_{ok\mathcal{R}}}\exp\left(-2\pi_{ok\mathcal{Y}}\beta_{gk}t_{gg}-2\pi_{ok\mathcal{R}}\alpha_{gk}t_{gg}\right)+\frac{2\pi_{okT}\pi_{okC}}{\pi_{ok\mathcal{Y}}}\exp\left(-2\pi_{ok\mathcal{R}}\beta_{gk}t_{gg}-2\pi_{ok\mathcal{Y}}\alpha_{gk}t_{gg}\right)\right\},\\
&&\mathcal{H}^\circ_{gg'}=1-\frac{1}{K}\sum_{k=1}^K\left\{\sum_{c=1}^C\pi_{okc}^2+\exp(-\beta_{gk}t_{gg}-\beta_{g'k}t_{g'g'})\left[\frac{\pi_{ok\mathcal{R}}}{\pi_{ok\mathcal{Y}}}\sum_{c'\in\mathcal{Y}}\pi^2_{okc'}+\frac{\pi_{ok\mathcal{Y}}}{\pi_{ok\mathcal{R}}}\sum_{c'\in\mathcal{R}}\pi^2_{okc'}\right]\right.\\
&+&\left.\frac{2\pi_{okA}\pi_{okG}}{\pi_{ok\mathcal{R}}}\exp\left(-\pi_{ok\mathcal{Y}}(\beta_{gk}t_{gg}+\beta_{g'k}t_{g'g'})-\pi_{ok\mathcal{R}}(\alpha_{gk}t_{gg}+\alpha_{g'k}t_{g'g'})\right)\right.\\
&+&\left.\frac{2\pi_{okT}\pi_{okC}}{\pi_{ok\mathcal{Y}}}\exp\left(-\pi_{ok\mathcal{R}}(\beta_{gk}t_{gg}+\beta_{g'k}t_{g'g'})-\pi_{ok\mathcal{Y}}(\alpha_{gk}t_{gg}+\alpha_{g'k}t_{g'g'})\right)\right\}.
\eeqn

\section{Asymptotic Theory}\label{teoAss}

Note that $D_n(t,B)$ is a generalized $U$-statistic, since $\bar{D}_{gg'}(t)$, $\bar{D}_{gg}(t)$ and $\bar{D}_{g'g'}(t)$ are all generalized $U$-statistics. We first study whether $D_n(t,B)$ is degenerated in the sense of Hoefdding (\citealp{Hoeffding:1948}). In this degenerated case, $D_n(t,B)$  will be a quasi $U$-statistic  (\citealp{Pinheiro:2009, Pinheiro:2011}). 

The asymptotic distribution of $D_n(t,B)$ under $H_0$ can be found considering the assumption of interchangeable sequences (\citealp{DeFinetti:1937}), which is an adequate assumption when we have an evolutionary process. Besides the depedence between sequences, another important issue is that of the dependence between-sites with-sequences, specially since genetic data is usually of high of ultra-high dimensions, i.e.. such that $K >> n$. One way of dealing with site dependence when $K>>n$ is considering mixing conditions (\citealp{Pinheiro:2009, Pinheiro:2011}).

We present below the conditions for first-order degeneracy.

\begin{proposition}\label{prepPI}
$ H_0: 2 \mathcal{H}^\circ_{gg'}-\mathcal{H}^\circ_{gg}-\mathcal{H}^\circ_{g'g'}=0$ is true if only if one the following two conditions holds: (i) $\pi_{g\mid o;lc'\mid c}(t_{gg})=\pi_{g'\mid o; lc'\mid c}(t_{g'g'})$, $\forall c'\in \mathcal{S}=\{A,C,G,T\}$ with $\pi_{olc}>0$ $\forall c\;\in \;\mathcal{S}$;  or (ii) $\pi_{olc}=1$ for some $c\;\in \;\mathcal{S}=\{A,C,G,T\}$, and $0$ otherwise.
\end{proposition}
\begin{proof}
Consider $H_0$ is true.  It follows from  (\ref{espDggl}) and (\ref{espDgg}) that 
\beqn
&&2\mathcal{H}^\circ_{gg'}-\mathcal{H}^\circ_{gg}-\mathcal{H}^\circ_{g'g'}=2\left\{1-\frac{1}{K}\sum_{l=1}^K\sum_{c=1}^C\pi_{olc}\sum_{c'=1}^C\pi_{g\mid o;lc' \mid c}(t_{gg})\pi_{g'\mid o; l c'\mid c}(t_{g'g'})\right\}+\\
&-&\left\{1-\frac{1}{K}\sum_{l=1}^K\sum_{c=1}^C\pi_{olc}\sum_{c'=1}^C\pi^2_{g\mid o;lc' \mid c}(t_{gg})\right\}-\left\{1-\frac{1}{K}\sum_{l=1}^K\sum_{c=1}^C\pi_{olc}\sum_{c'=1}^C\pi^2_{g'\mid o;lc' \mid c}(t_{g'g'})\right\}\\
&=&\frac{1}{K}\sum_{l=1}^K\sum_{c=1}^C\pi_{olc}\sum_{c'=1}^C(\pi_{g\mid o;lc'\mid c}(t_{gg})-\pi_{g'\mid o; l c'\mid c}(t_{g'g'}))^2.
\eeqn
Analyzing the terms of this result, we have to
\begin{enumerate}
\item
If for some $c\in\mathcal{S}$, $\pi_{olc}=1$ and $\pi_{old}=0$, for $d\neq c$, then
\beqn
\frac{1}{K}\sum_{l=1}^K\sum_{c'=1}^C(\pi_{g\mid o;lc'\mid c}(t_{gg})-\pi_{g'\mid o; l c'\mid c}(t_{g'g'}))^2=0.
\eeqn 
Then, $\pi_{g\mid o;lc'\mid c}(t_{gg})=\pi_{g'\mid o; l c'\mid c}(t_{g'g'})$, $\forall c'\in\mathcal{S}$, with $\pi_{olc}=1$ for $c\in\mathcal{S}$.
\item
Under HKY model, the probabilities for each nucleotide are different, $\pi_{olc}>0$ $\forall c\in\mathcal{S}$, with $\sum_{c=1}^C\pi_{olc}=1$,
\beqn
\frac{1}{K}\sum_{l=1}^K\sum_{c=1}^C\pi_{olc}\sum_{c'=1}^C(\pi_{g\mid o;lc'\mid c}(t_{gg})-\pi_{g'\mid o; lc'\mid c}(t_{g'g'}))^2=0.
\eeqn
For this equality to be true, $\pi_{g\mid o;lc'\mid c}(t_{gg})=\pi_{g'\mid o; lc'\mid c}(t_{g'g'})$, $\forall c'\in\mathcal{S}$, with  $\pi_{olc}>0$ and $\forall c\in\mathcal{S}$.
\end{enumerate}
Thus,
$\pi_{g\mid o;lc'\mid c}(t_{gg})=\pi_{g'\mid o; lc'\mid c}(t_{g'g'}) \forall c'\in \mathcal{S} \mbox{ and } l=1,\ldots, K$, 
when $\pi_{olc}>0$. Then, $\pi_{g\mid o;lc'\mid c}(t_{gg})=\pi_{g'\mid o; lc'\mid c}(t_{g'g'})$  is a necessary and sufficient condition to $H_0$ be true. 
\end{proof}

Theorem \ref{theo31} establishes that $D_n(t,B)$ is degenerated under homogeneity among groups for the coalescence processes..

\begin{theorem}\label{theo31}
Consider $H_0: 2\mathcal{H}^\circ_{gg'}-\mathcal{H}^\circ_{gg}-\mathcal{H}^\circ_{g'g'}=0$. Under the coalescence process, $D_n(t,B)$ given by (\ref{dntb_equation}) is first order degenerated in Hoeffding's sense.
\end{theorem}

\begin{proof}
First, assume that the MRCA sequence between groups $g$ and $g'$ occur at time $t_{ogg'}$. Consider $\phi(X_1,X_2)$ the kernel  of a U-statistics. The Hoeffding's decomposition is given by 
\begin{eqnarray}
\phi(X_1,X_2)=\phi_0+\Psi_1(X_1)+\Psi_1(X_2)+\Psi_2(X_1,X_2),
\end{eqnarray}
where $\Psi_1(X_1)=\{\phi_1(X_1)-\phi_0\}$ (similarly to $\Psi_1(X_2)$), 
$\Psi_2(X_1,X2)=\{\phi(X_1,X2)-\phi_1(X_1)-\phi_1(X_2)+\phi_0\}$,
\[
\phi_0=E[\phi(X_1,X_2)]\;,\;\phi_1(X_1)=E[\phi(X_1,X_2)\mid X_1]\;\mbox{and}\;\phi_1(X_2)=E[\phi(X_1,X_2)\mid X_2].
\]
The first projection for $\bar{D}_{gg}(t)$ is given by
\[
\phi_g(\mathbf{X}_{gi}) = E\left[\frac{1}{K}\sum_{l=1}^K\mathbb{I}(X_{gil}(t)\neq X_{gjl}(t))\mid X_{gil}(t)\right].
\]
Consider $\mathbf{d}\in \mathcal{S}^K$, that is, a vector of $K$ positions, representing a possible sequence
\begin{eqnarray*}
\phi_g(\mathbf{c'})&=&\frac{1}{K}\sum_{l=1}^K\sum_{c=1}^C\pi_{olc}\sum_{c'=1}^CP(X_{gjl}(t)\neq X_{gil}(t)\mid X_{gil}(t)=c', X_{ol}(t_{ogg'})=c) \\
&=&1-\frac{1}{K}\sum_{l=1}^K\sum_{c=1}^C\pi_{olc}\sum_{c'=1}^C\pi_{g \mid o; l c'\mid c}(t_{gg}).
\end{eqnarray*}
So,
\beqn
\phi_g(\mathbf{X}_{gi})&=&1-\frac{1}{K}\sum_{l=1}^K\sum_{c=1}^C\pi_{olc}\pi_{g \mid o; l X_{gil}(t)\mid c}(t_{gg}),
\eeqn
where $\pi_{g \mid o; l X_{gil}(t)\mid c}(t_{gg})=P(X_{gil}(t)=c'\mid X_{ol}(t_{ogg'})=c)$, for $\forall c' \in \mathcal{S}$.

The term related to the variance of the first Hoeffding's decomposition is given by
\begin{eqnarray}
\zeta_1 &=& E\left[\phi^2_g(\mathbf{X}_{gi})\right]-\left(\mathcal{H}^\circ_{gg}\right)^2 \nonumber \\
&=&\frac{1}{K^2} \sum_{k=1}^K\sum_{l=1}^K\sum_{c=1}^C\sum_{d=1}^C\pi_{okc}\pi_{old}\sum_{c'=1}^C\sum_{d'=1}^C\pi_{g \mid o; k c'\mid c}(t_{gg})\pi_{g \mid o; l d'\mid d}(t_{gg}) w_{g \mid o; kc', l d'\mid c, d}(t_{gg}),\label{varzeta}
\end{eqnarray}
where $w_{g \mid o; kc', l d'\mid c, d}(t_{gg})=\pi_{g \mid o; kc', l d'\mid c, d}(t_{gg})-\pi_{g \mid o; k c'\mid c}(t_{gg})\pi_{g \mid o; l d'\mid d}(t_{gg})$.

Under the mixing conditions (\citealp{Withers:1981, Doukhan:1994}), $w_{g \mid o; kc', l d'\mid c, d}(t_{gg})\rightarrow 0$ as $K\rightarrow \infty$  (which ensures the covariance term between the sites is summable). Some conditions can be imposed on $w_{g \mid o; kc', l d'\mid c, d}(t_{gg})$ to be negligible:
\begin{enumerate}
\item Auto Regressive dependency (AR), with $w_{g\mid o; kc',l d'\mid c,d}(t_{gg})=|\pi_{g\mid o; kc',l d'\mid c,d}(t_{gg})-\pi_{g \mid o; k c'\mid c}(t_{gg})$   
$\pi_{g \mid o; l d'\mid d}(t_{gg})|=\rho^{|k-l|}$, with $0<\rho<1$, i.e., the more distant the sites are, $\lim_{r\rightarrow \infty}\rho^{r}=0$ (\citealp{Doukhan:1994}).
\item $m$-dependency, with $\pi_{g\mid o; kc',ld' \mid c,d}(t_{gg})-\pi_{g \mid o; k c'\mid c}(t_{gg})\pi_{g \mid o; l d'\mid d}(t_{gg})=0$, for $|k-l|>m$.
\end{enumerate}
Considering mixing conditions, $\bar{D}_{gg}(t)$ is degenerated under $H_0$, because $\zeta_1$ (given by (\ref{varzeta})) goes to 0.

$\bar{D}_{gg'}$ is a generalized U-statistics of degree (1,1) and we can write 
\beqn
\phi_{10(g')}(\mathbf{X}_{gi})&=&1-\frac{1}{K}\sum_{l=1}^K\sum_{c=1}^C\pi_{olc}\pi_{g' \mid o; l X_{gil}(t)\mid c}(t_{g'g'}),\\
\phi_{01(g)}(\mathbf{X}_{g'j})&=&1-\frac{1}{K}\sum_{l=1}^K\sum_{c=1}^C\pi_{olc}\pi_{g\mid o; l X_{g'jl}(t)\mid c}(t_{gg}).
\eeqn
The terms $\zeta_{10}$ and $\zeta_{01}$ are given by
\[
\zeta_{10}=E\left[\phi^2_{10(g')}(\mathbf{X}_{gi})\right]-\left(\mathcal{H}^\circ_{gg'}\right)^2\;\;\;\mbox{and}\;\;\;
\zeta_{01}=E\left[\phi^2_{01(g)}(\mathbf{X}_{g'i})\right]-\left(\mathcal{H}^\circ_{gg'}\right)^2.
\]
So,
\[
\zeta_{10}
=\frac{1}{K^2}\sum_{k=1}^K\sum_{l=1}^K\sum_{c=1}^C\sum_{d=1}^C\pi_{okc}\pi_{old}\sum_{c'=1}^C\sum_{d'=1}^C\pi_{g' \mid o; k c'\mid c}(t_{g'g'})\pi_{g' \mid o; l d'\mid d}(t_{g'g'})w_{g\mid o; k c',ld'\mid c,d}(t_{gg}),
\]
where $w_{g\mid o; k c';ld'\mid c,d}(t_{gg})=\pi_{g\mid o; k c',ld'\mid c,d}(t_{gg})-\pi_{g\mid o; k c'\mid c}(t_{gg})\pi_{g\mid o; l d'\mid d}(t_{gg})\rightarrow 0$ as $K\rightarrow \infty$. Analogously, for $\zeta_{01}$. So, the terms of the $D_n(t,B)$ are degenerated under $H_0$.

We can write, $\bar{D}_{gg}(t)$ and $\bar{D}_{gg'}(t)$ as 
\beqn
\bar{D}_{gg}(t)&=&\mathcal{H}^{\circ_{gg}}+\frac{2}{n_g}\sum_{i=1}^{n_g}\Psi_{1g}(\mathbf{X}_{gi})+\binom{n_g}{2}^{-1} \sum_{i \le j}^{1,n_g} \Psi_{2g}(\mathbf{X}_{gi},\mathbf{X}_{gj}),
\eeqn
where $\Psi_{1g}(\mathbf{X}_{g\cdot})=\left(\phi_{g}(\mathbf{X}_{g\cdot})-\mathcal{H}^\circ_{gg}\right)$ and 
\beqn
\Psi_{2g}(\mathbf{X}_{gi},\mathbf{X}_{gj})=\frac{1}{K}\sum_{l=1}^K \mathbb{I}(X_{gil}(t) \neq X_{gjl}(t))-\phi_{g}(\mathbf{X}_{gi})-\phi_{g}(\mathbf{X}_{gj})+\mathcal{H}^\circ_{gg}.
\eeqn
Also,
\begin{eqnarray*}
\bar{D}_{gg'}(t)&=&\mathcal{H}^\circ_{gg'}+\frac{1}{n_g}\sum_{i=1}^{n_{g}}\Psi_{10g'}(\mathbf{X}_{gi})+\frac{1}{n_{g'}}\sum_{j=1}^{n_{g'}}\Psi_{01g}(\mathbf{X}_{g'j})+\frac{1}{n_gn_{g'}}\sum_{i,j}^{n_g,n_{g'}}\Psi_{11g'g}(\mathbf{X}_{gi},\mathbf{X}_{g'j}),
\end{eqnarray*}
where $\Psi_{10g'}(\mathbf{X}_{gi}) = \left ( \phi_{10(g')}(\mathbf{X}_{gi})-\mathcal{H}_{gg'}^\circ\right )$, $\Psi_{01g}(\mathbf{X}_{g'j}) = \left(\phi_{01(g)}(\mathbf{X}_{g'j})-\mathcal{H}_{gg'}^\circ\right)$ and 
\begin{eqnarray*}
\Psi_{11g'g}(\mathbf{X}_{gi},\mathbf{X}_{g'j})&=&\frac{1}{K}\sum_{l=1}^K \mathbb{I}(X_{gil}(t) \neq X_{g'jl}(t))-\phi_{10(g')}(\mathbf{X}_{gi})-\phi_{01(g)}(\mathbf{X}_{g'j})+\mathcal{H}^\circ_{gg'}.
\end{eqnarray*}

Using these results, we can rewrite $D_n(t,B)$ as
\begin{eqnarray}
D_n(t,B)&=&\frac{1}{n(n-1)}\sum_{g=1}^{G-1}\sum_{g'=g+1}^G n_g n_{g'}\left\{2\mathcal{H}^\circ_{gg'}-\mathcal{H}^\circ_{gg}-\mathcal{H}^\circ_{g'g'}+\frac{2}{n_g}\sum_{i=1}^{n_{g}}\left(\Psi_{10g'}(\mathbf{X}_{gi})-\Psi_{1g}(\mathbf{X}_{gi})\right)+\right. \nonumber \\
&+&\frac{2}{n_{g'}}\sum_{j=1}^{n_{g'}}\left(\Psi_{01g}(\mathbf{X}_{g'j})-\Psi_{1g'}(\mathbf{X}_{g'i})\right)+\frac{2}{n_gn_{g'}}\sum_{i,j}^{n_g,n_{g'}}\Psi_{11g'g}(\mathbf{X}_{gi},\mathbf{X}_{g'j})+ \nonumber \\
&-&\left.{\binom{n_g}{2}}^{-1} \sum_{i \le j}^{1,n_g} \Psi_{2(g)}(\mathbf{X}_{gi},\mathbf{X}_{gj})-\binom{n_{g'}}{2}^{-1} \sum_{i \le j}^{1,n_{g'}} \Psi_{2(g')}(\mathbf{X}_{g'i},\mathbf{X}_{g'j})\right\}.
\end{eqnarray}

Under $H_0: 2\mathcal{H}^\circ_{gg'}-\mathcal{H}^\circ_{gg}-\mathcal{H}^\circ_{g'g'}=0$, we have
\begin{enumerate}
\item
$\pi_{g\mid o;kd \mid c}(t_{gg})=\pi_{g'\mid o; k d\mid c}(t_{g'g'})=\pi_{\cdot\mid o,kd\mid c}(t^\circ)$, $\forall d\in \mathcal{S}$ with $\pi_{okc}>0$, where $t^{\circ}$ is the maximum of the times until the MRCA to all the sequences, or
\item
$\pi_{okc}=1$ for some $c\in \mathcal{S}$ and $\pi_{okc}=0$, otherwise.
\end{enumerate}

Let
\beqn
&&\Psi^{\star}(\mathbf{X}_{gi})=\frac{2}{n_g}\sum_{i=1}^{n_{g}}\left(\Psi_{10g'}(\mathbf{X}_{gi})-\Psi_{1g}(\mathbf{X}_{gi})\right)=\mathcal{H}^\circ_{gg}-\mathcal{H}_{gg'}^\circ+\\
&+&\frac{2}{n_g}\sum_{i=1}^{n_{g}}\left(\frac{1}{K}\sum_{k=1}^K\sum_{c=1}^C\pi_{okc}\pi_{g \mid o; k X_{gik}\mid c}(t_{gg})-\frac{1}{K}\sum_{k=1}^K\sum_{c=1}^C\pi_{okc}\pi_{g' \mid o; k X_{gik}\mid c}(t_{g'g'})\right)
\eeqn
and
\begin{eqnarray}
&&\Psi^{\star}(\mathbf{X}_{g'j})=\frac{2}{n_{g'}}\sum_{j=1}^{n_{g'}}\left(\Psi_{01g}(\mathbf{X}_{g'j})-\Psi_{1g'}(\mathbf{X}_{g'i})\right)=\mathcal{H}^\circ_{g'g'}-\mathcal{H}_{gg'}^\circ+ \nonumber \\
&+&\frac{2}{n_{g'}}\sum_{j=1}^{n_{g'}}\left(\frac{1}{K}\sum_{k=1}^K\sum_{c=1}^C\pi_{okc}\pi_{g' \mid o; k X_{g'jk}\mid c}(t_{g'g'})-\frac{1}{K}\sum_{k=1}^K\sum_{c=1}^C\pi_{okc}\pi_{g\mid o; k X_{g'jk}\mid c}(t_{gg})\right). \nonumber
\end{eqnarray}
Under $H_0$, $\pi_{g' \mid o; k X_{gik}\mid c}(t_{g'g'})=\pi_{g \mid o; k X_{gik}\mid c}(t_{gg})=\pi_{\cdot\mid o,kd\mid c}(t^\circ)$ for $\mathbf{X}_{gi}\in$ $\mathcal{S}^K$ and $\pi_{g' \mid o; k X_{g'jk}\mid c}(t_{g'g'})=\pi_{g \mid o; k X_{g'jk}\mid c}(t_{gg})=\pi_{\cdot\mid o,kd\mid c}(t^\circ)$ for $\mathbf{X}_{g'j}\in$ $\mathcal{S}^K$. Then,
\beqn
\Psi^{\star}(\mathbf{X}_{gi})+\Psi^{\star}(\mathbf{X}_{g'j})=\mathcal{H}^\circ_{gg}+\mathcal{H}^\circ_{g'g'}-2\mathcal{H}_{gg'}^\circ=0.
\eeqn
Therefore, under $H_0$
\beqnn
D_n(t,B)&=&\frac{1}{n(n-1)}\sum_{g=1}^{G-1}\sum_{g'=g+1}^G\frac{1}{K}\sum_{k=1}^K n_g n_{g'}\left\{\frac{2}{n_gn_{g'}}\sum_{i,j}^{n_g,n_{g'}}\Psi_{11(g',g)}(\mathbf{X}_{gi},\mathbf{X}_{g'j})+\right.\nonumber\\
&-&\left.\binom{n_{g}}{2}^{-1} \sum_{i \le j}^{1,n_g}\Psi_{2(g)}(\mathbf{X}_{gi},\mathbf{X}_{gj})-
\binom{n_{g'}}{2}^{-1} \sum_{i \le j}^{1,n_{g'}}\Psi_{2(g')}(\mathbf{X}_{g'i},\mathbf{X}_{g'j})\right\}.\label{dnbH0}
\eeqnn
As the term $\Psi^{\star}_{g'}(\mathbf{X}_{gi})+\Psi^{\star}_{g}(\mathbf{X}_{g'j})$, which corresponds to the first-order Hoeffding's decomposition of the test statistic, is identically equal to 0 under $H_0$, then $D_n(t,B)$ is degenerate of first order.
\end{proof}

Rewriting $D_n(t,B)$ in terms of the second-order Hoeffding's decomposition, we obtain a quasi U-statistics (\citealp{Pinheiro:2009}). Let $\mathbf{X}_1, \ldots, \mathbf{X}_n$, be random vectors representing the DNA sequences, where the first $n_1$ indexes are related to group $1$, the $n_2$ to group $2$, $\ldots$, and the last $n_G$ to group $G$. Thus, the quasi U-statistic is given by
\beqn
T_n=\binom{n}{2}^{-1}\sum_{i<j}^{1,n} \eta_{nij} \Psi(\mathbf{X}_i(t),\mathbf{X}_j(t)),
\eeqn
where the weights are given by
\beqnn
\label{pesoestatU}
\eta_{n,i,j}=\begin{cases}
1, & \mbox{if $i$ and $j$  come from different groups}\\
-\frac{n-n_g}{n_g-1}, & \mbox{if $i$ and $j$ are both from the same group $g$, $1 \leq g \leq G$}\\
\end{cases}
\eeqnn
and
\begin{equation}
\Psi(\mathbf{X}_{i}(t),\mathbf{X}_{j}(t))=\frac{1}{K}\sum_{k=1}^K \mathbb{I}(X_{ik}(t) \neq X_{jk}(t))-\phi_{10}(\mathbf{X}_{i})-\phi_{01}(\mathbf{X}_{j})+\mathcal{H}^\circ_{\cdot\cdot},
\end{equation}
where $\mathcal{H}^\circ_{\cdot\cdot} = \frac{1}{K}\sum_{k=1}^KE\left(\mathbb{I}(X_{ik}(t) \neq X_{jk}(t))\right)$ and 
$\phi_{10}(\mathbf{X}_i) = E\left[\frac{1}{K}\sum_{k=1}^K\mathbb{I}(X_{ik}(t)\neq X_{jk}(t))\mid \mathbf{X}_{i}(t)\right]$. The properties of the weights can be evaluated in \citep{Pinheiro:2009}.

The following Theorem shows that the test statistic has an asymptotic Normal distribution when $K$ or $n$ are large under some conditions, like interchangeable random vectors (\citealp{DeFinetti:1937}) and the mixing conditions (\citealp{Withers:1981, Doukhan:1994}). A collection of random variables is defined to be interchangeable if every finite subcollection has a joint distribution which is a symmetric function of its arguments. Infinite sequences of exchangeable random variables can be regarded equivalently as sequences of conditionally i.i.d random variables.

\begin{theorem} \label{teo15}
Suppose that under $H_0$ the random vectors are interchangeable and conditionally independent and identically distributed given the MRCA. Suppose the following mixing condition
\beqnn
\sum_{1\le l <k\le K}E\left[\phi^{\star}(X_{i_1l}, \ldots,X_{i_ml})\phi^{\star}(X_{i_1k}, \ldots,X_{i_mk})\right]=O(K)\mbox{ when }K\rightarrow \infty \label{condmixing_quase}
\eeqnn
is valid. Suppose also, $E[\Psi^2(X_i(t),X_j(t))]>0$. If $T_n$ is a quasi U-statistic,
\beqn
T_n /\sqrt{Var(T_n)} \xrightarrow{D} N(0,1),
\eeqn
as $K\rightarrow \infty$ (either if $n\rightarrow \infty$, $n/K\rightarrow 0$, as $K\rightarrow \infty$ or if $n$ is bounded).
\end{theorem}
\begin{proof}
We apply Theorem 2.1 from \citet{Withers:1981}, which states the asymptotic distribution of sum of i.i.d random vectors, with dependence on its components. We are adding the assumption of interchangeable random vectors, since we have dependency of sequences.
 
Consider two individuals, $i$ and $j$, of group $g$. Considering that $\mathbf{X}_{o}(t_{ogg})$ is the ancestral sequence
\beqn
&&P(\mathbf{X}_{gi}(t)=\mathbf{c},\mathbf{X}_{gj}(t)=\tilde{\mathbf{c}}\mid \mathbf{X}_{o}(t_{ogg})=\mathbf{c}^\star)\\
&& =P(\mathbf{X}_{gi}(t)=\mathbf{c}\mid \mathbf{X}_{o}(t_{ogg})=\mathbf{c}^\star)P(\mathbf{X}_{gj}(t)=\tilde{\mathbf{c}}\mid \mathbf{X}_{o}(t_{ogg})=\mathbf{c}^\star),
\eeqn
where $\mathbf{c}$, $\tilde{\mathbf{c}}$ and $\mathbf{c}^\star$ are vectors with possible elements of $\mathbf{X}_{gi}(t)$, $\mathbf{X}_{gj}(t)$ and $\mathbf{X}_{o}(t_{ogg})$, respectively. In group $g$, the vectors $\mathbf{X}_{gi}(t)$, $i=1,\ldots, n_g$, will have the same distribution, which will depend on the group parameter. 

Thus, the random vectors $\mathbf{X}_{gi}(t)$ $i=1, \ldots, n_g$ are interchangeable. Similarly, considering two individuals $i$ and $j$ of distinct groups. 

The variance of $T_n$ can be found using the assumptions that $o$ is the MRCA of all sequences, past the time $t^\circ$ ($t^\circ=\max\{t_i-t_0\}$), where $ t_i $, $i=1, \ldots n$ is the total length of the branch of sequence $i$ to the ancestral sequence $o$. Recalling that under $H_0$ the sequences are assumed to be interchangeable,
\beqn
&&Var(T_n)=\binom{n}{2}^{-2}\left\{E\left[Var\left(\sum_{i<j}^{1,n} \eta_{nij}\Psi(\mathbf{X}_{i}(t),\mathbf{X}_{j}(t))\mid \mathbf{X}_{o}(t_0)\right)\right]+\right.\nonumber \\
&+&\left. Var\left[E\left(\sum_{i<j}^{1,n} \eta_{nij}\Psi(\mathbf{X}_{i}(t),\mathbf{X}_{j}(t))\mid \mathbf{X}_{o}(t_0)\right)\right]\right\} \nonumber \\
&=&\binom{n}{2}^{-2}\left\{\sum_{i<j}^{1,n} \eta^2_{nij}E\left[Var\left(\Psi(\mathbf{X}_{i}(t),\mathbf{X}_{j}(t))\mid \mathbf{X}_{o}(t_0)\right)\right]+\right. \nonumber \\
&+&\left. Var\left[E\left(\sum_{i<j}^{1,n} \eta_{nij} \Psi(\mathbf{X}_{i}(t),\mathbf{X}_{j}(t))\mid \mathbf{X}_{o}(t_0)\right)\right]\right\}.
\eeqn
Note that under $H_0$, we have $\pi_{g'\mid o; k d\mid c}(t_{g'g'})=\pi_{g\mid o; k d\mid c}(t_{gg})=\pi_{\cdot\mid o; k d\mid c}(t^\circ)$. The first moment is given by
\beqn
&&\frac{1}{K}\sum_{k=1}^KE\left(\mathbb{I}(X_{ik}(t) \neq X_{jk}(t))\right)=1-\frac{1}{K}\sum_{k=1}^K\sum_{c=1}^C\pi_{okc}\sum_{c'=1}^C\pi^2_{\cdot\mid o;kc'\mid c}(t^\circ)=\mathcal{H}^\circ_{\cdot\cdot}
\eeqn
and
\begin{equation}
E\left(\phi_{10}(\mathbf{X}_{i})\mid \mathbf{X}_{o}(t_0)\right)=\mathcal{H}^\circ_{\cdot\cdot}.
\end{equation}
Thus,
\beqn
&&E\left(\sum_{i<j}^{1,n} \eta_{nij}\Psi(\mathbf{X}_{i}(t),\mathbf{X}_{j}(t))\mid \mathbf{X}_{o}(t_0)\right)=\\
&=&\sum_{i<j}^{1,n} \eta_{nij}E\left(\frac{1}{K}\sum_{k=1}^K \mathbb{I}(X_{ik}(t) \neq X_{jk}(t))-\phi_{10}(\mathbf{X}_{i})-\phi_{01}(\mathbf{X}_{j})+\mathcal{H}^\circ_{\cdot\cdot}\mid \mathbf{X}_{o}(t_0)\right)\\
&=&\sum_{i<j}^{1,n} \eta_{nij}\left[E\left(\frac{1}{K}\sum_{k=1}^K \mathbb{I}(X_{ik}(t) \neq X_{jk}(t))\mid \mathbf{X}_{o}(t_0)\right)-\mathcal{H}^\circ_{\cdot\cdot}\right]\\
&=&\sum_{i<j}^{1,n} \eta_{nij}\left[1-\frac{1}{K}\sum_{k=1}^K\sum_{c'=1}^C\pi^2_{\cdot\mid o;kc'\mid \mathbf{X}_{o}(t_0)}(t^\circ)-\mathcal{H}^\circ_{\cdot\cdot}\right]=0.
\eeqn
Using the assumption that the sequences are interchangeable, i.e., the random vectors representing the sequences are independent and identically distributed when conditioned on the random vector of the MRCA sequence, the first term of the variance of $T_n$ is given by
\beqn
&&E\left[Var\left(\sum_{i<j}^{1,n} \eta_{nij}\Psi(\mathbf{X}_{i}(t),\mathbf{X}_{j}(t))\mid \mathbf{X}_{o}(t_0)\right)\right]=\sum_{i<j}^{1,n} \eta^2_{nij}E\left[Var\left(\frac{1}{K}\sum_{k=1}^K \mathbb{I}(X_{ik}(t) \neq X_{jk}(t))\right]\right.\\
&&-\left. (\phi_{10}(\mathbf{X}_{i})-\mathcal{H}^\circ_{\cdot\cdot})-(\phi_{01}(\mathbf{X}_{j})-\mathcal{H}^\circ_{\cdot\cdot})-\mathcal{H}^\circ_{\cdot\cdot}\mid \mathbf{X}_{o}(t_0)\right].
\eeqn
By Hoeffding's decomposition, the terms are orthogonal and have the following results,
\beqn
&&E\left[Var\left(\frac{1}{K}\sum_{k=1}^K \mathbb{I}(X_{ik}(t) \neq X_{jk}(t))-\mathcal{H}^\circ_{\cdot\cdot}\mid \mathbf{X}_{o}(t_0)\right)\right]= \nonumber \\
&=&\frac{2}{K}\sum_{k=1}^K\sum_{c=1}^C\pi_{okc}\sum_{c'=1}^C\pi^2_{\cdot \mid o;kc'\mid c}(t^\circ)-\frac{2}{K^2}\left\{\sum_{k=1}^K\sum_{l=1}^K\sum_{c=1}^C\sum_{d=1}^C\pi_{okc,ld}\sum_{c'=1}^C\pi^2_{\cdot \mid o;kc'\mid c}(t^\circ)\sum_{d'=1}^C\pi^2_{\cdot \mid o;ld'\mid d}(t^\circ)+\right. \nonumber \\
&-&\left.\frac{1}{2}\sum_{k=1}^K\sum_{l=1}^K\sum_{c=1}^C\sum_{d=1}^C\pi_{okc,ld}\sum_{c'=1}^C\sum_{d'=1}^C\left[\pi^2_{\cdot\mid o; kc', ld'\mid c,d}(t^\circ)-\pi^2_{\cdot \mid o;kc'\mid c}(t^\circ)\pi^2_{\cdot \mid o;ld'\mid d}(t^\circ)\right]\right\},\mbox{ and}
\eeqn
\[
E\left[Var\left((\phi_{01}(\mathbf{X}_{j})-\mathcal{H}^\circ_{\cdot\cdot})\mid \mathbf{X}_{o}(t_0)\right)\right]=E\left\{E\left(\phi^2_{01}(\mathbf{X}_{j})\mid \mathbf{X}_{o}(t_0)\right)-\left[E\left(\phi_{01}(\mathbf{X}_{j})\mid \mathbf{X}_{o}(t_0)\right)\right]^2\right\}.
\]
The terms of $\phi_{01}$ and $\phi_{10}$, the probability was conditioned in terms of its ancestral sequence at time $t_{ogg'}$. With no loss of generality, suppose that the ancestral sequence occurs at $t_0$. When $H_0$ is true, we have the following results
\beqnn
&&E\left[Var\left((\phi_{01}(\mathbf{X}_{j})-\mathcal{H}^\circ_{\cdot\cdot})\mid \mathbf{X}_{o}(t_0)\right)\right]=E\left[E\left(1-\frac{1}{K}\sum_{k=1}^K\sum_{c=1}^C\pi_{okc}\pi_{\cdot \mid o;kX_{jk}\mid c}(t^\circ)\right)^2\right]-\left(\mathcal{H}^\circ_{\cdot\cdot}\right)^2\nonumber\\
&&=1-\frac{2}{K}\sum_{k=1}^K\sum_{c=1}^C\pi_{okc}\sum_{c'=1}^C\pi^2_{\cdot \mid o;kc''\mid c}(t^\circ)+\nonumber\\
&&+\frac{1}{K^2}\sum_{k=1}^K\sum_{l=1}^K\sum_{c=1}^C\sum_{d=1}^C\pi_{okc}\pi_{old}\sum_{c'=1}^C\sum_{d'=1}^C\pi_{\cdot \mid o;kc'\mid c}(t^\circ)\pi_{\cdot \mid o;ld'\mid d}(t^\circ)\pi_{\cdot \mid o;kc',ld'\mid c,d}(t^\circ)-\left(\mathcal{H}^\circ_{\cdot\cdot}\right)^2\nonumber\\
&&=\frac{1}{K^2}\sum_{k=1}^K\sum_{l=1}^K\sum_{c=1}^C\sum_{d=1}^C\pi_{okc}\pi_{old}\sum_{c'=1}^C\sum_{d'=1}^C\pi_{\cdot \mid o;kc'\mid c}(t^\circ)\pi_{\cdot \mid o;ld'\mid d}(t^\circ)\left[\pi_{\cdot \mid o;kc',ld'\mid c,d}(t^\circ)\right.+\nonumber\\
&&\left.-\pi_{\cdot \mid o;kc'\mid c}(t^\circ)\pi_{\cdot \mid o;kd'\mid d}(t^\circ)\right]\label{varphi0110}.
\eeqnn
An analogous result to (\ref{varphi0110}) can be obtained for $E\left[Var\left((\phi_{10}(\mathbf{X}_{i})-\mathcal{H}^\circ_{\cdot\cdot})\mid \mathbf{X}_{o}(t_0)\right)\right]$. Thus,
\begin{eqnarray}
&&E\left[Var\left(\sum_{i<j}^{1,n} \eta_{nij}\Psi(\mathbf{X}_{i}(t),\mathbf{X}_{j}(t))\mid \mathbf{X}_{o}(t_0)\right)\right]= \nonumber \\
&=&\sum_{i<j}^{1,n} \eta^2_{nij}\left\{-\frac{1}{K^2}\sum_{k=1}^K\sum_{l=1}^K\sum_{c=1}^C\sum_{d=1}^C\pi_{okc,ld}\sum_{c'=1}^C\sum_{d'=1}^C\left[\pi^2_{\cdot\mid o; kc', ld'\mid c,d}(t^\circ)-\pi^2_{\cdot \mid o;kc'\mid c}(t^\circ)\pi^2_{\cdot \mid o;ld'\mid d}(t^\circ)\right]\right.+\nonumber \\
&+&\frac{2}{K}\sum_{k=1}^K\sum_{c=1}^C\pi_{okc}\sum_{c'=1}^C\pi^2_{\cdot \mid o;kc'\mid c}(t^\circ)-\frac{2}{K^2}\sum_{k=1}^K\sum_{l=1}^K\sum_{c=1}^C\sum_{d=1}^C\pi_{okc,ld}\sum_{c'=1}^C\sum_{d'=1}^C\pi^2_{\cdot ``\mid o;kc'\mid c}(t^\circ)\pi^2_{\cdot \mid o;ld'\mid d}(t^\circ)+ \nonumber \\
&+&\frac{2}{K^2}\sum_{k=1}^K\sum_{l=1}^K\sum_{c=1}^C\sum_{d=1}^C\pi_{okc}\pi_{old}\sum_{c'=1}^C\sum_{d'=1}^C\pi_{\cdot \mid o;kc'\mid c}(t^\circ)\pi_{\cdot \mid o;ld'\mid d}(t^\circ)\left[\pi_{\cdot \mid o;kc',ld'\mid c,d}(t^\circ)+\right. \nonumber \\
&-&\left.\left.\pi_{\cdot \mid o;kc'\mid c}(t^\circ)\pi_{\cdot \mid o;kd'\mid d}(t^\circ)\right]\right\}. \nonumber 
\end{eqnarray}
By doing some manipulations in order to define the terms that one can apply the mixing conditions, we have

\begin{eqnarray}
&&E\left[Var\left(\sum_{i<j}^{1,n} \eta_{nij}\Psi(\mathbf{X}_{i}(t),\mathbf{X}_{j}(t))\mid \mathbf{X}_{o}(t_0)\right)\right]=\nonumber \\
&=&\sum_{i<j}^{1,n} \eta^2_{nij}\left\{-\frac{1}{K^2}\sum_{k=1}^K\sum_{l=1}^K\sum_{c=1}^C\sum_{d=1}^C[\pi_{okc,ld}-\pi_{okc}\pi_{old}]\sum_{c'=1}^C\sum_{d'=1}^C\left[\pi_{\cdot\mid o; kc', ld'\mid c,d}(t^\circ)+\right.\right. \nonumber \\
&-&\left.\pi_{\cdot \mid o;kc'\mid c}(t^\circ)\pi_{\cdot \mid o;ld'\mid d}(t^\circ)\right]^2-\frac{2}{K^2}\sum_{k=1}^K\sum_{l=1}^K\sum_{c=1}^C\sum_{d=1}^C[\pi_{okc,ld}-\pi_{okc}\pi_{old}]\sum_{c'=1}^C\sum_{d'=1}^C\pi_{\cdot \mid o;kc'\mid c}(t^\circ)\times \nonumber \\
&\times&\pi_{\cdot \mid o;ld'\mid d}(t^\circ)+\frac{1}{K^2}\sum_{k=1}^K\sum_{l=1}^K\sum_{c=1}^C\sum_{d=1}^C\pi_{okc}\pi_{old}\sum_{c'=1}^C\sum_{d'=1}^C\left[\pi_{\cdot\mid o; kc', ld'\mid c,d}(t^\circ)+\right. \nonumber \\
&-&\left.\pi_{\cdot \mid o;kc'\mid c}(t^\circ)\pi_{\cdot \mid o;ld'\mid d}(t^\circ)\right]^2+\frac{2}{K^2}\sum_{k=1}^K\sum_{l=1}^K\sum_{c=1}^C\sum_{d=1}^C\pi_{okc}\pi_{old}\sum_{c'=1}^C\sum_{d'=1}^C\pi_{\cdot \mid o;kc'\mid c}(t^\circ)\pi_{\cdot \mid o;ld'\mid d}(t^\circ)+ \nonumber \\
&+&\frac{2}{K}\sum_{k=1}^K\sum_{c=1}^C\pi_{okc}\sum_{c'=1}^C\pi^2_{\cdot \mid o;kc'\mid c}(t^\circ)-\frac{2}{K^2}\sum_{k=1}^K\sum_{l=1}^K\sum_{c=1}^C\sum_{d=1}^C[\pi_{okc,ld}-\pi_{okc}\pi_{old}]\sum_{c'=1}^C\sum_{d'=1}^C\pi^2_{\cdot \mid o;kc'\mid c}(t^\circ)\times \nonumber \\
&\times&\pi^2_{\cdot \mid o;ld'\mid d}(t^\circ)+\frac{2}{K^2}\sum_{k=1}^K\sum_{l=1}^K\sum_{c=1}^C\sum_{d=1}^C\pi_{okc}\pi_{old}\sum_{c'=1}^C\sum_{d'=1}^C\pi^2_{\cdot \mid o;kc'\mid c}(t^\circ)\pi^2_{\cdot \mid o;ld'\mid d}(t^\circ)+ \nonumber 
\end{eqnarray} 
\begin{eqnarray}
&+&\frac{2}{K^2}\sum_{k=1}^K\sum_{l=1}^K\sum_{c=1}^C\sum_{d=1}^C\pi_{okc}\pi_{old}\sum_{c'=1}^C\sum_{d'=1}^C\pi_{\cdot \mid o;kc'\mid c}(t^\circ)\pi_{\cdot \mid o;ld'\mid d}(t^\circ)\left[\pi_{\cdot \mid o;kc',ld'\mid c,d}(t^\circ)+\right. \nonumber \\
&-&\left.\left.\pi_{\cdot \mid o;kc'\mid c}(t^\circ)\pi_{\cdot \mid o;kd'\mid d}(t^\circ)\right]\right\}. \nonumber
\end{eqnarray}
Now consider the following equations:
\beqn
w_{okc,ld}&=&\pi_{okc,ld}-\pi_{okc}\pi_{old}\mbox{, and}\\
w_{\cdot\mid o; kc',ld'\mid c,d} (t^\circ)&=&\pi_{\cdot\mid o; kc',ld'\mid c,d}(t^\circ)-\pi_{\cdot\mid o; kc'\mid c} (t^\circ)\pi_{\cdot\mid o; ld'\mid d}(t^\circ).
\eeqn
If $M_n=\sum_{i<j}^{1,n} \eta^2_{nij}$, then
\begin{eqnarray*}
&&Var(T_n)=\binom{n}{2}^{-2}M_n\left\{-\frac{1}{K^2}\sum_{k=1}^K\sum_{l=1}^K\sum_{c=1}^C\sum_{d=1}^Cw_{okc,ld}\sum_{c'=1}^C\sum_{d'=1}^Cw^2_{\cdot\mid o; kc',ld'\mid c,d} (t^\circ)+\right.\\
&&-\frac{2}{K^2}\sum_{k=1}^K\sum_{l=1}^K\sum_{c=1}^C\sum_{d=1}^Cw_{okc,ld}\sum_{c'=1}^C\sum_{d'=1}^C\pi_{\cdot \mid o;kc'\mid c}(t^\circ)\pi_{\cdot \mid o;ld'\mid d}(t^\circ)+\frac{1}{K^2}\sum_{k=1}^K\sum_{l=1}^K\sum_{c=1}^C\sum_{d=1}^C\pi_{okc}\pi_{old}\nonumber\\
&&\times\sum_{c'=1}^C\sum_{d'=1}^Cw^2_{okc,ld}+\frac{2}{K^2}\sum_{k=1}^K\sum_{l=1}^K\sum_{c=1}^C\sum_{d=1}^C\pi_{okc}\pi_{old}\sum_{c'=1}^C\sum_{d'=1}^C\pi_{\cdot \mid o;kc'\mid c}(t^\circ)\pi_{\cdot \mid o;ld'\mid d}(t^\circ)\\
&&+\frac{2}{K}\sum_{k=1}^K\sum_{c=1}^C\pi_{okc}\sum_{c'=1}^C\pi^2_{\cdot \mid o;kc'\mid c}(t^\circ)-\frac{2}{K^2}\sum_{k=1}^K\sum_{l=1}^K\sum_{c=1}^C\sum_{d=1}^Cw_{okc,ld}\sum_{c'=1}^C\sum_{d'=1}^C\pi^2_{\cdot \mid o;kc'\mid c}(t^\circ)\\
&&\times\pi^2_{\cdot \mid o;ld'\mid d}(t^\circ)+\frac{2}{K^2}\sum_{k=1}^K\sum_{l=1}^K\sum_{c=1}^C\sum_{d=1}^C\pi_{okc}\pi_{old}\sum_{c'=1}^C\sum_{d'=1}^C\pi^2_{\cdot \mid o;kc'\mid c}(t^\circ)\pi^2_{\cdot \mid o;ld'\mid d}(t^\circ)\\
&&\left.+\frac{2}{K^2}\sum_{k=1}^K\sum_{l=1}^K\sum_{c=1}^C\sum_{d=1}^C\pi_{okc}\pi_{old}\sum_{c'=1}^C\sum_{d'=1}^C\pi_{\cdot \mid o;kc'\mid c}(t^\circ)\pi_{\cdot \mid o;ld'\mid d}(t^\circ)w_{\cdot\mid o; kc',ld'\mid c,d}(t^\circ)\right\}, \label{varTn_sobH0}
\end{eqnarray*}
where $M_n=O(n^2)$. When $K$ is fixed, $Var(T_n)=O(n^{-2})$ and when $n\rightarrow\infty$, $nT_n=O_p(1)$. When $n$ is fixed and $K\rightarrow\infty$, the mixing conditions must be assumed in terms of $w_{okc,ld}$ and $w_{\cdot\mid o;kc',ld'\mid c,d}(t^\circ)$. Therefore,
\beqn
&&\!\!\!\!\!\!Var(T_n)= M_nK^{-1}\binom{n}{2}^{-2}\left\{W_{K}+\frac{2}{K}\sum_{k=1}^K\sum_{c=1}^C\pi_{okc}\sum_{c'=1}^C\pi_{\cdot \mid o;kc'\mid c}(t^\circ)\sum_{l=1}^K\sum_{d=1}^C\pi_{old}\sum_{d'=1}^C\pi_{\cdot \mid o;ld'\mid d}(t^\circ) \nonumber\right. \\
&+&\left.2\sum_{k=1}^K\sum_{c=1}^C\pi_{okc}\sum_{c'=1}^C\pi^2_{\cdot \mid o;kc'\mid c}(t^\circ)+\frac{2}{K}\sum_{k=1}^K\sum_{c=1}^C\pi_{okc}\sum_{c'=1}^C\pi^2_{\cdot \mid o;kc'\mid c}(t^\circ)\sum_{l=1}^K\sum_{d=1}^C\pi_{old}\sum_{d'=1}^C\pi^2_{\cdot \mid o;ld'\mid d}(t^\circ)\right\},
\eeqn
where $W_K$ is representing the dependence of the sites. So, the convergence is of order $\sqrt{K} $, when $K\rightarrow \infty$ and $\sqrt{K}T_n=O_p(1)$. When $K,n\rightarrow \infty$, $Var(T_n)=O(n^{-2}K^{-1})$, it means that, the Central Limit Theorem holds with convergence rate of the order $n\sqrt{K}$.
\end{proof}

Under $H_0$ and the stationary coalescence process, the sequences can also be considered interchangeable and the asymptotic results follow using  Theorem \ref{teo15}.

It is also important to find the asymptotic distribution under $H_1$. One way to find it is considering artificial situation where the distance between $H_0$ and $H_1$ decreases when $n\rightarrow \infty$ in order of $n^{-1/2}$ as we have in Pitman local alternatives (\citealp{Pitman:1949}).

Considering Proposition \ref{prepPI} which states that under $H_0$, $2\mathcal{H}^\circ_{gg'}-\mathcal{H}_{gg}^\circ-\mathcal{H}_{g'g'}^\circ=0$ implies that $\pi_{g\mid o;k\mathbf{c'}\mid \mathbf{c}}(t_{gg})=\pi_{g'\mid o;k c'\mid c}(t_{g'g'})=\pi_{\cdot\mid o;k c'\mid c}(t^\circ)$, where  $t^\circ=\max\{t_{gg},t_{g'g'}\}$, $g,g'=1, \ldots,G$ and the parameter $a_{n,g}(\cdot)$ as an approximation of $H_0$ under $H_{1n}$ for $K$ fixed
\beqn
H_{1n}: \pi_{g\mid o;kc'\mid c}(t^{\circ}) =\pi_{\cdot\mid o;kc'\mid c}(t^{\circ})+a_{n,g}(c'), \;\;c',c\in \mathcal{S}^K,
\eeqn
such that $|a_{n,g}(c')|\rightarrow 0$, when $n\rightarrow \infty$ and \\
\beqn
&&\sum_{c'=1}^{C}\left(\pi_{\cdot\mid o;kc'\mid c}(t^{\circ})+a_{n,g}(c')\right)=1\implies \sum_{c'=1}^{C}a_{n,g}(c')=0.
\eeqn

By Theorem \ref{teo15}, $nT_nV^{-1/2}\xrightarrow{D} N(0,1)$, with $Var(T_n)=V/n^2$, when $K$ is fixed. Under $H_{1n}$, consider that $E(T_n)=\mu_n$. Then, asymptotic  normality occurs by the Theorem, considering that
\beqn
n(T_n-\mu_n)V^{-1/2}+n\mu_nV^{-1/2}=\sqrt{n}(T_n-\mu_n)\sqrt{nV^{-1}}+n\mu_nV^{-1/2}\xrightarrow{D} N(0,1),
\eeqn
with convergence of order $n^{1/2}$, since it  is not degenerate under the alternative hypothesis. 
Now, suppose $\mu_n=\gamma/n$. Then, the term does not depend on  $n$ and  it is bounded by $\gamma$ and $V^{-1/2}$. In this context, the Pitman alternative is given by 
\begin{equation}
H_{1n}: \pi_{g\mid o;kc'\mid c}(t^{\circ}) =\pi_{\cdot\mid o;kc'\mid c}(t^{\circ})+\frac{\gamma_g(c')}{n}, \;\;c',c\in \mathcal{S}^K,\label{PitmanLocal}
\end{equation}
where $\gamma_g(c')\neq 0$ is fixed and 
\[
\sum_{c'=1}^{C}\left(\pi_{\cdot\mid o;kc'\mid c}(t^{\circ})+\frac{\gamma_g(c')}{n}\right)=1\implies \sum_{c'=1}^{C}\frac{\gamma_g(c')}{\sqrt{n}}=0.
\]
With Pitman alternatives defined, one verifies the contiguity of the probability measure of $H_{1n}$ under $H_0$. Let $(\Omega_n,\mathcal{F}_n)$, $n\ge 1$, be sequences of measurable spaces
and $P_n$, $Q_n$ are probability measures defined in $(\Omega_n,\mathcal{F}_n)$. 
The sequence $\{Q_n\}$ is said to be contiguous to $\{P_n\}$ if
\beqn
\lim_{n\to \infty}P_n(A_n)=0,\; A_n\in \mathcal{F}_n\;\implies \lim_{n\to \infty}Q_n(A_n)=0.
\eeqn
Suppose $\{P_n\}$ and $\{Q_n\}$ are absolutely continuous with respect to the Lebesgue measure in  $(\Omega_n,\mathcal{F}_n)$, with
\beqn
p_n=\frac{dP_n}{d\mu_n}\;\;\mbox{and}\;\;q_n=\frac{dQ_n}{d\mu_n}.
\eeqn
Now, define the likelihood ratio by
\beqn
L_n=\begin{cases}
{q_n}/{p_n} & \quad \mbox{if }p_n>0,\\
1&\quad \mbox{if } q_n=p_n=0\quad \mbox{and}\\
\infty&\quad \mbox{if } q_n>p_n=0.
\end{cases}
\eeqn

Lemma \ref{lemma_31} establishes contiguity for Pitman alternatives.

\begin{lem}\label{lemma_31}
Considering the Pitiman local alternatives $H_{1n}$ in (\ref{PitmanLocal}), the sequences of $H_{1n}$ are contiguous to $H_0$.
\end{lem}
\begin{proof}
Considering the functions in terms of $H_{1n}$ and $H_0$, then
\beqn
L_n=\begin{cases}
{f(H_{1n})}/{f(H_0)} & \quad \mbox{if }f(H_0)>0,\\
1&\quad \mbox{if } f(H_{1n})=f(H_0)=0\quad \mbox{and}\\
\infty&\quad \mbox{if } f(H_{1n})>f(H_0)=0.
\end{cases}
\eeqn
The likelihood function under $H_{1n}$ is given by
\begin{equation}
f(H_{1n})=\prod_{g=1}^{G}\frac{n_{g}!}{\prod_{c'\in \mathcal{S}^K}n_{g}(c')}\prod_{c'\in \mathcal{S}^K}\pi_{g\mid o;kc'\mid c}(t^{\circ})^{n_{g}(c')}
\end{equation}
and, under $H_0$,
\beqn
f(H_0)=\prod_{g=1}^{G} \frac{n_{g}!}{\prod_{c'\in \mathcal{S}^K}n_{g}(c')}\prod_{c'\in \mathcal{S}^K}\pi_{\cdot\mid o;kc'\mid c}(t^{\circ})^{n_{g}(c')}.
\eeqn
Considering the case where $f(H_0)>0$, the logarithm of the likelihood ratio statistic  is given by 

\begin{eqnarray}
\log L_n&=&\log\left(\frac{f(H_{1n})}{f(H_0)}\right)=\log\left(\frac{\prod_{g=1}^{G} \frac{n_{g}!}{\prod_{c'\in \mathcal{S}^K}n_{g}(c')}\prod_{c'\in \mathcal{S}^K}\pi_{g\mid o;kc'\mid c}(t^{\circ})^{n_{g}(c')}}{\prod_{g=1}^{G}\frac{n_{g}!}{\prod_{c'\in \mathcal{S}^K}n_{g}(c')}\prod_{c'\in \mathcal{S}^K}\pi_{\cdot\mid 
o;kc'\mid c}(t^{\circ})^{n_{g}(c')}}\right) \nonumber \\
&=&\sum_{g=1}^{G}\left[\log\left(\frac{n_{g}!}{\prod_{c'\in \mathcal{S}^K}n_{g}(c')}\right)+\sum_{c'\in \mathcal{S}^K}n_{g}(c')\log(\pi_{g\mid o;kc'\mid c}(t^{\circ}))\right]+ \nonumber 
\end{eqnarray}
\begin{eqnarray}
&-&\sum_{g=1}^{G}\left[\log\left(\frac{n_{g}!}{\prod_{c'\in \mathcal{S}^K}n_{g}(c')}\right)+\sum_{c'\in \mathcal{S}^K}n_{g}(c')\log(\pi_{\cdot\mid o;kc'\mid c}(t^{\circ}))\right] \nonumber \\
&=&\sum_{g=1}^{G}\sum_{c'\in \mathcal{S}^K}n_{g}(c')\log\left(\frac{\pi_{g\mid o;kc'\mid c}(t^{\circ})}{\pi_{\cdot\mid o;kc'\mid c}(t^{\circ})}\right)=\sum_{g=1}^{G}\sum_{c'\in \mathcal{S}^K}n_{g} (c')\log\left(1+\frac{1}{n}\frac{\gamma_g(c')}{\pi_{\cdot\mid o;kc'\mid c}(t^{\circ})}\right).
\end{eqnarray}
If  $N_{g}(c')$ is a random variable representing the number of elements in group $g$ in  category $c'$, the asymptotic distribution of $N_{g}(c')$ is Normal, under $H_0$, i.e., 
\beqn
\frac{\sqrt{n_g}(N_g(c')/n_g-\pi_{\cdot\mid o;kc'\mid c}(t^{\circ}))}{\sqrt{\pi_{\cdot\mid o;kc'\mid c}(t^{\circ})(1-\pi_{\cdot\mid o;kc'\mid c}(t^{\circ}))}}
\xrightarrow{\mathcal{D}} N(0,1), 
\eeqn 
when $n_g=\min\{n_1,\ldots, n_G\}\rightarrow \infty ~~(g=1, \ldots, G)$. But,
\beqn
\log\left(1+\frac{1}{n}\frac{\gamma_g(c')}{\pi_{\cdot\mid o;kc'\mid c}(t^{\circ})}\right) \rightarrow 0\;\; \mbox{almost sure, when } n\rightarrow \infty. 
\eeqn
Then, by  Slutsky Theorem $\log L_n\xrightarrow{\mathcal{D}}0$. When $L_n>0$, $L_n \xrightarrow{\mathcal{D}}1$ and $E_{H_0}(L)=1$, $H_{1n}$ is contiguous with respect to $H_0$, i.e., the  local alternatives in (\ref{PitmanLocal}) are contiguous by Le Cam's First Lemma (\citealp{Lecam}).

Therefore, the distribution under $H_{1n}$ is normal considering the Pitman local alternative  ({\ref{PitmanLocal}}).
\end{proof}

\begin{theorem} \label{teoAlternativaH1}
Under the following conditions:  
\begin{itemize}
\item
(a) the random vectors are exchangeable  and conditionally independent given the common ancestror; 
\item
(b) the mixing condition given in (\ref{condmixing_quase}) is true,
\item
(c) The local Pitman alternative given by 
\beqnn
H_{1nK}: \pi_{g\mid o;kc'\mid c}(t^{\circ}) =\pi_{\cdot\mid o;kc'\mid c}(t^{\circ})+\frac{\gamma_g(c')}{n\sqrt{K}}, \;\;c',c\in \mathcal{S}^K\label{HnKdef}
\eeqnn
\item
(d) $0<E[\phi^2(\mathbf{X}_{gi},\mathbf{X}_{g'j})]<\infty$,
\item
(e) contiguity of  $H_{1n}$.
\end{itemize}
Then, 
\beqn
\sqrt{nK}\left(D_n(t,B)-\Delta\right)V^{-1/2}\xrightarrow{\mathcal{D}} N(0,1),\mbox{ when $K,n\rightarrow \infty$,}
\eeqn
where $\Delta=({n^3(n-1)K})^{-1}\sum_{g>g'}\sum_{c'=1}^C{n_g n_{g'}\gamma^2_{gg'}(c')}$ and  $Var(D_n(t,B))=V/n^2K$, where $V$ is computed considering $H_{nK}$ and (\ref{varTn_sobH0}) adding to  $\pi_{\cdot\mid o;kc'\mid c}(t^{\circ})$ the term ${\gamma_g(c')}/{n\sqrt{K}}$.
\end{theorem}
\begin{proof}
Consider the Pitman local alternatives given by (\ref{HnKdef}) and let 
\beqn
\Delta=\frac{1}{n(n-1)}\sum_{g>g'}n_gn_{g'}\left(2\mathcal{H}^\circ_{gg'}-\mathcal{H}_{gg}^\circ-\mathcal{H}_{g'g'}^\circ\right)=E(D_n(t,B)).
\eeqn
Writing $\Delta$ as a function of the transition probabilities, 
\beqn
&&\Delta=\frac{1}{n(n-1)}\sum_{g>g'}n_gn_{g'}	\left\{2\left[1-\frac{1}{K}\sum_{k=1}^K\sum_{c=1}^C\pi_{okc}\sum_{c'=1}^C\pi_{g\mid o;kc' \mid c}(t_{gg})\pi_{g'\mid o; k c'\mid c}(t_{g'g'})\right]+\right.\\
&-&\left.\left[1-\frac{1}{K}\sum_{k=1}^K\sum_{c=1}^C\pi_{okc}\sum_{c'=1}^C\pi^2_{g\mid o;kc' \mid c}(t_{gg})\right]-\left[1-\frac{1}{K}\sum_{k=1}^K\sum_{c=1}^C\pi_{okc}\sum_{c'=1}^C\pi^2_{g'\mid o;kc' \mid c}(t_{g'g'})\right]\right\}\\
&=&\frac{1}{n(n-1)}\sum_{g>g'}n_gn_{g'}\frac{1}{K}\sum_{k=1}^K\sum_{c=1}^C\pi_{okc}\sum_{c'=1}^C(\pi_{g\mid o;kc' \mid c}(t_{gg})-\pi_{g'\mid o;kc' \mid c}(t_{g'g'}))^2,
\eeqn
Under $H_{1nK}$ given by (\ref{HnKdef}), 
\beqn
\Delta&=&\frac{1}{n(n-1)}\sum_{g>g'}n_gn_{g'}\frac{1}{K}\sum_{k=1}^K\sum_{c=1}^C\pi_{okc}\sum_{c'=1}^C\left(\pi_{\cdot\mid o;kc'\mid c}(t^{\circ})+\frac{\gamma_g(c')}{n\sqrt{K}}-\pi_{\cdot\mid o;kc'\mid c}(t^{\circ})-\frac{\gamma_{g'}(c')}{n\sqrt{K}}\right)^2\\
&=&\frac{1}{n(n-1)}\sum_{g>g'}n_gn_{g'}\sum_{c'=1}^C\frac{(\gamma_g(c')-\gamma_{g'}(c'))^2}{n^2K}=\frac{1}{n(n-1)}\sum_{g>g'}\sum_{c'=1}^C\frac{n_gn_{g'}\gamma^2_{gg'}(c')}{n^2K},
\eeqn
then the Pitman local alternative is given by $H^\star_{1nK}:(\gamma_{gg'}(c')/{n^2K})>0$ for some  $c'\in \mathcal{S}$. 
Therefore, 
\beqn
\sqrt{n}\sqrt{nK}\left(D_n(t,B)-\Delta\right)V^{-1/2}+n\sqrt{K}\frac{1}{n(n-1)}\sum_{g>g'}\sum_{c'=1}^C\frac{n_gn_{g'}\gamma^2_{gg'}(c')}{n^2K}V^{-1/2}=B_{nK_1}+B_{nK_2},
\eeqn
where $Var(D_n(t,B))=V/n^2K$ and $V$ is computed considering $H_{1nK}$ using (\ref{varTn_sobH0}). By Theorem \ref{teo15} $B_{nK_1}\xrightarrow{\mathcal{D}} N(0,1)$ when $K$ or $n$ goes to infinity,  while  $B_{nK_2}\xrightarrow{p}0$. Under contiguous local alternatives, the asymptotic distribution  of  $D_n(t,B)$ is Normal.
\end{proof}

\section{Simulation Study}\label{Simul}

In order to verify the asymptotic distribution of $D_n(t,B)$ under $H_0$, we generated DNA sequences using Monte Carlo simulations under each substitution model discussed in \citet{Tavare:2004} (\citealp{Jukes:1969}, \citealp{Kimura:1980}, \citealp{Felsenstein:1981} and \citealp{Hasegawa:1985}). A phylogenetic tree is generated following Figure \ref{figcoalseq2}(a)  using the package {\sf phyclust} from {\sf R} available on CRAN (\url{https://cran.r-project.org/web/packages/available_packages_by_name.html}). 
For instance, under the Jukes-Cantor model (JC69) it is assumed that, given a nucleotide, the substitutions occur at rate $\alpha_l$ for $l=1,\ldots,K$, for all nucleotides in $\mathcal{S}=\{A,C,G,T\}$ and  the relative frequency of nucleotides are all equal with $\pi_{lc}=1/4$ for $\forall c\in \mathcal{S}=\{A,C,G,T\}$. Then,  for $l=1,\ldots, K$, 
$\pi_{lc'\mid c }(t) = 1/4 -[exp(-4\lambda_lt)]/4$ and $\pi_{lc\mid c }(t)= 1/4 + [3exp(-4\lambda_lt)]/4$. For the HKY model, substitutions occur at rate $\alpha_l$ (for transitions), $\beta_l$ (for transversions) and $\pi_{lc}$ varies. So, for the HKY model, we have the relations given by (\ref{NucTHKY})-(\ref{NucGHKY}).
It was considered $G = 2, 3$ and  $4$ groups in the simulations.



For each substitution model (JC69, K80, F81 and HKY85), four populations were generated, 
with $N$ sequences in each population group ($N = 25000$) and each sequence with a number of fixed sites equal to $K = (100, 2000)$. To generate the populations we used the  package {\sf phyclust} from {\sf R} available on CRAN  (\citealp{Chen:2011}), which includes the open source programs {\sf ms} for generation of coalescent trees (\citealp{Hudson:2002})  and {\sf seq-gen} for generation of DNA sequences (\citealp{Rambaut:1997}). {\sf Seq-Gen} routine from 
\citet{Rambaut:1997} simulates the evolution of nucleotide sequences along a genealogy tree, using substitution models for DNA sequences. Under $H_0$, a coalescent tree with $4N$ individuals/sequences is generated. Then, there are 4 populations of size $N$ with parameters listed in Table \ref{tabModelsParameters}, which values were based on  \citet{Hasegawa:1985} study 
for human mitochondrial DNA.

\begin{table} [h!]
\begin{center}
\caption{Model parameters used in the simulation data under $H_0$.}
\label{tabModelsParameters}
  \begin{tabular}{c c c c c c c}
	\hline
  \hline
Models &transition& transversion& $\pi_A$& $\pi_C$&$\pi_G$ & $\pi_T$\\
\hline
JC69 & $10^{-09}$ & - & $0.25$&$0.25$&$0.25$&$0.25$\\
K80 & $10^{-09}$ & $20^{-09}$ & $0.25$&$0.25$&$0.25$&$0.25$\\
F81& $10^{-09}$ & -  &$0.36$&$0.43$&$0.04$&$0.17$\\
HKY85& $10^{-09}$ & $20^{-09}$ &$0.36$&$0.43$&$0.04$&$0.17$\\
\hline
\hline
  \end{tabular}
	\end{center}
\end{table}

The goal is to study the distribution of $D_n(B)$ under $H_0$: homogeneity among groups. 
For this purpose, we took samples without replacement of size $n=n_1+\ldots + n_G$, where the first $n_1$ are from group 1, the next $n_2$ from group 2 and so on according to the number of groups ($G=2, 3, 4$).
Consider $S$ the number of simulations, the algorithm is as follows: \\
For $s=1, \ldots, S$
\begin{enumerate}
\item Randomly choose a population group;
\item From this population group draw randomly $n=(n_1, \ldots, n_G)$ sequences;
\item The first $n_1$ sequences of the sample will be assigned to group 1, from $n_1$ to $n_1+n_2$ sequences to group 2 and so on.
\item
Compute $D_n(t,B)_{obs}$ (observed) from the sample $n=(n_1, \ldots, n_G)$.
\end{enumerate}
Repeat steps 1 to 4 $S$ times.

Under $H_1$, the transition and transversion parameters were different for each population (see Table \ref{tabModelsParametersH1}). 

\begin{table} [htb!]
\begin{center}
\caption{Transition and Transversion parameters for each population under $H_1$.}
\label{tabModelsParametersH1}
  \begin{tabular}{c c c c c}
  \hline
Parameters&$G_1$&$G_2$&$G_3$&$G_4$\\
\hline
transition & $1\times 10^{-08}$&$4\times 10^{-08}$&$1.6\times 10^{-07}$&$3.2\times 10^{-07}$\\
transversion & $2\times 10^{-08}$&$8\times 10^{-08}$&$3.2\times 10^{-07}$&$6.4\times 10^{-07}$\\
\hline
\hline
  \end{tabular}
	\end{center}
\end{table}

 For all simulations, unbalanced and balanced cases were considered with Q-Q plots to verify if the distribution of $D_n(t, B)$ is close to Normal.  Figures  \ref{fig:QQHKYG2}, \ref{fig:QQHKYG3} and \ref{fig:QQHKYG4} show Normal Q-Q plots for samples of $G=2, 3$  and $G=4$, respectively, with $K=2000$, assuming HKY85 model. 

We see deviations from normality, especially for unbalanced cases. 
In Figures \ref{fig:QQHKYG2}, \ref{fig:QQHKYG3} and \ref{fig:QQHKYG4} 
we can see deviations from normality for all situations and no evidence of Normality with 5\% level of significance was found. The best cases are the balanced ones with $n=100$ and $K=2000$.

\begin{figure}[htb]
\centering
\subfloat{%
\includegraphics[width=7.5cm,height=5.5cm]{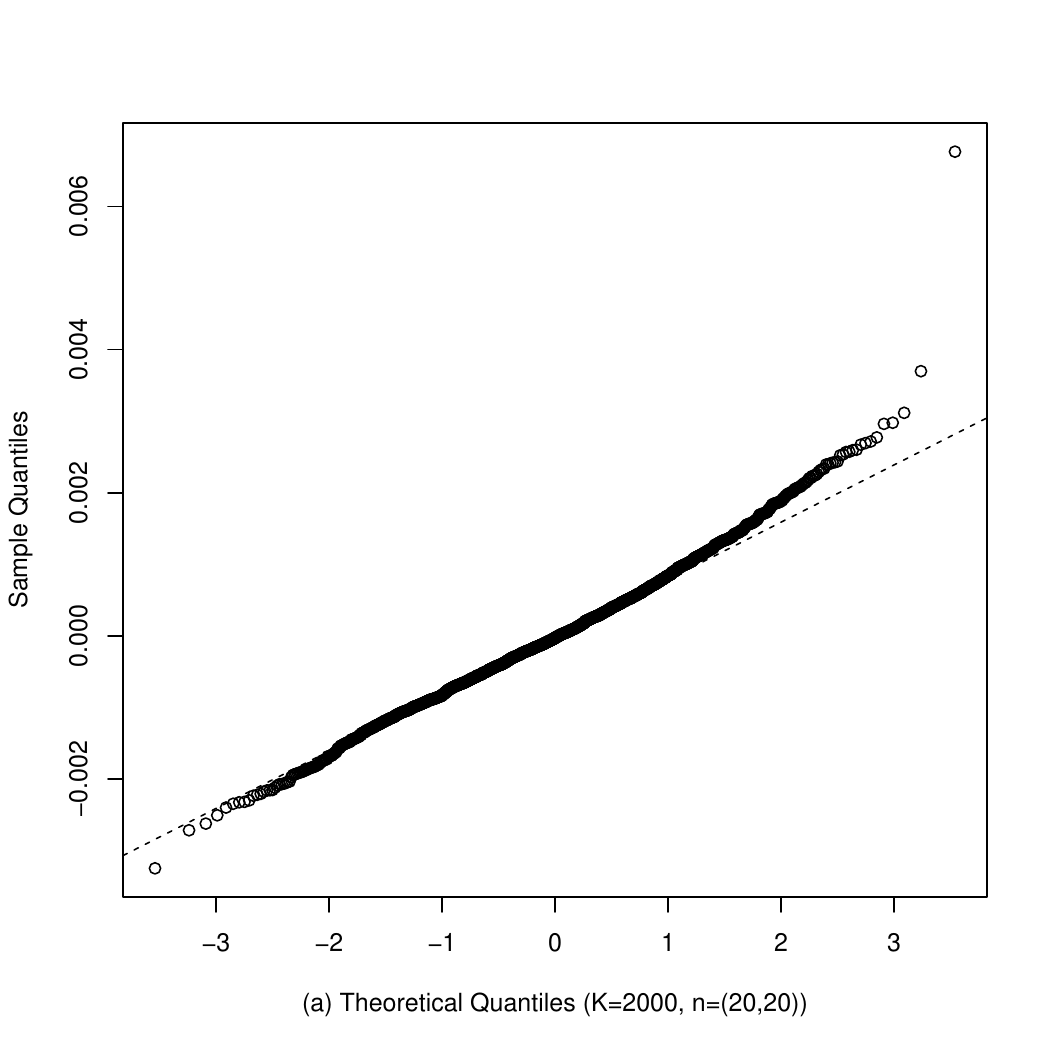}
}
\quad
\subfloat{%
\includegraphics[width=7.5cm,height=5.5cm]{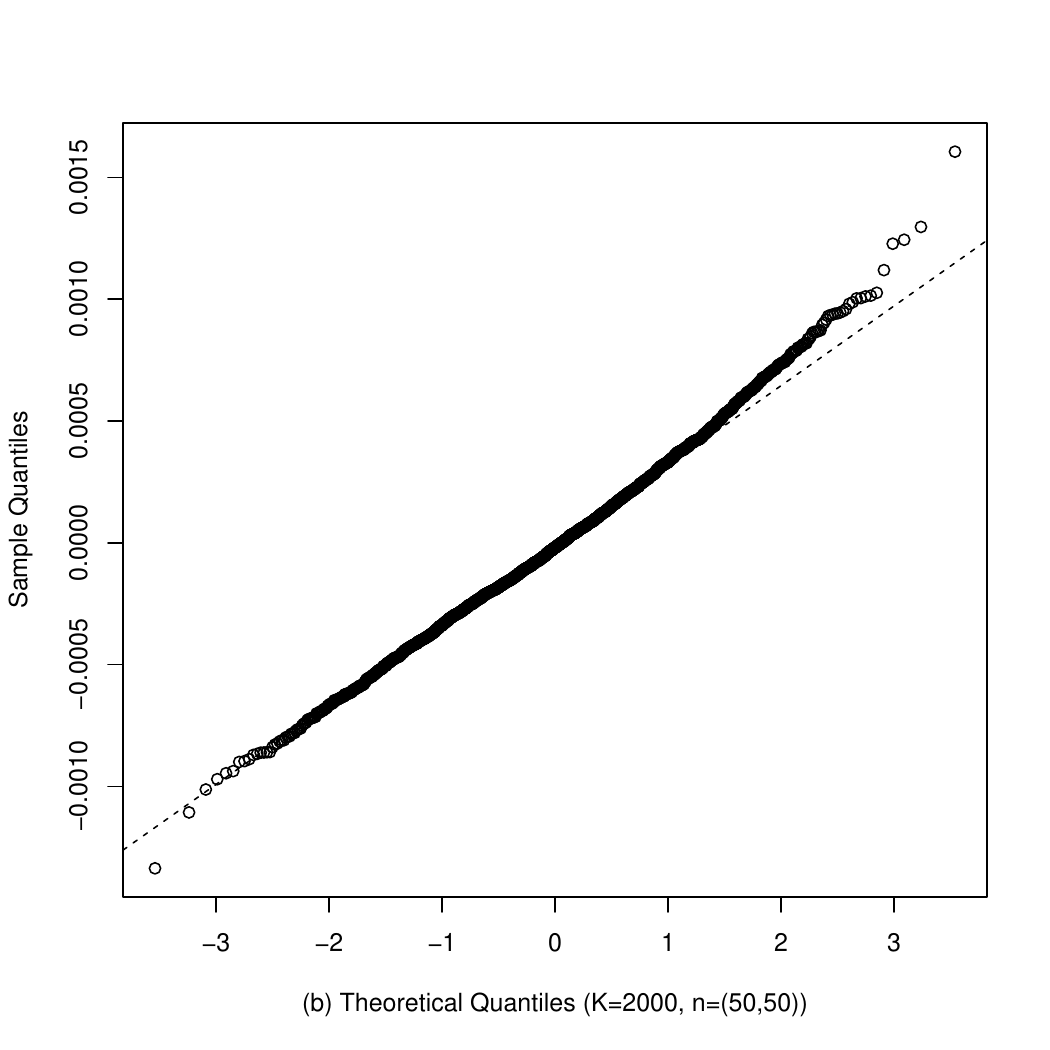}
}

\subfloat{%
\includegraphics[width=7.5cm,height=5.5cm]{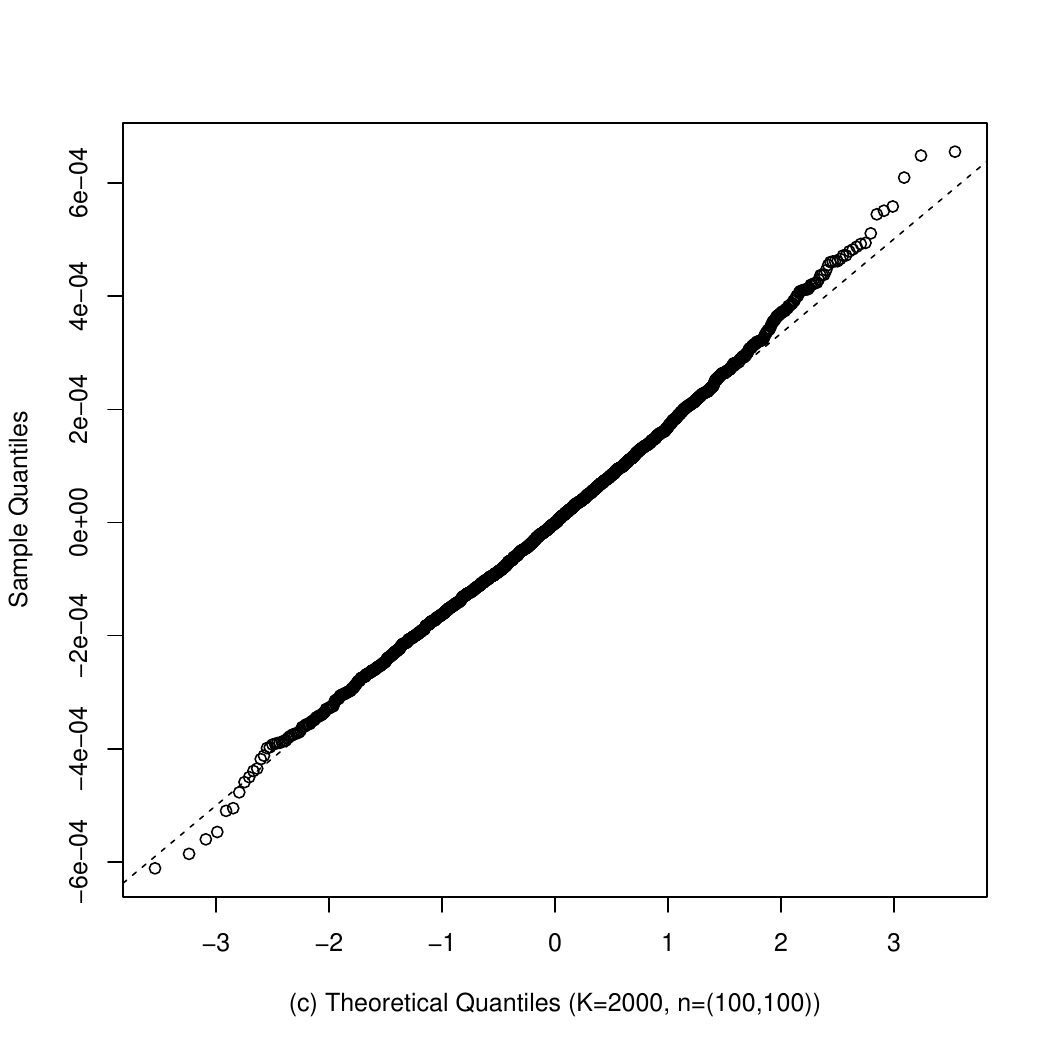}
}
\quad
\subfloat{%
\includegraphics[width=7.5cm,height=5.5cm]{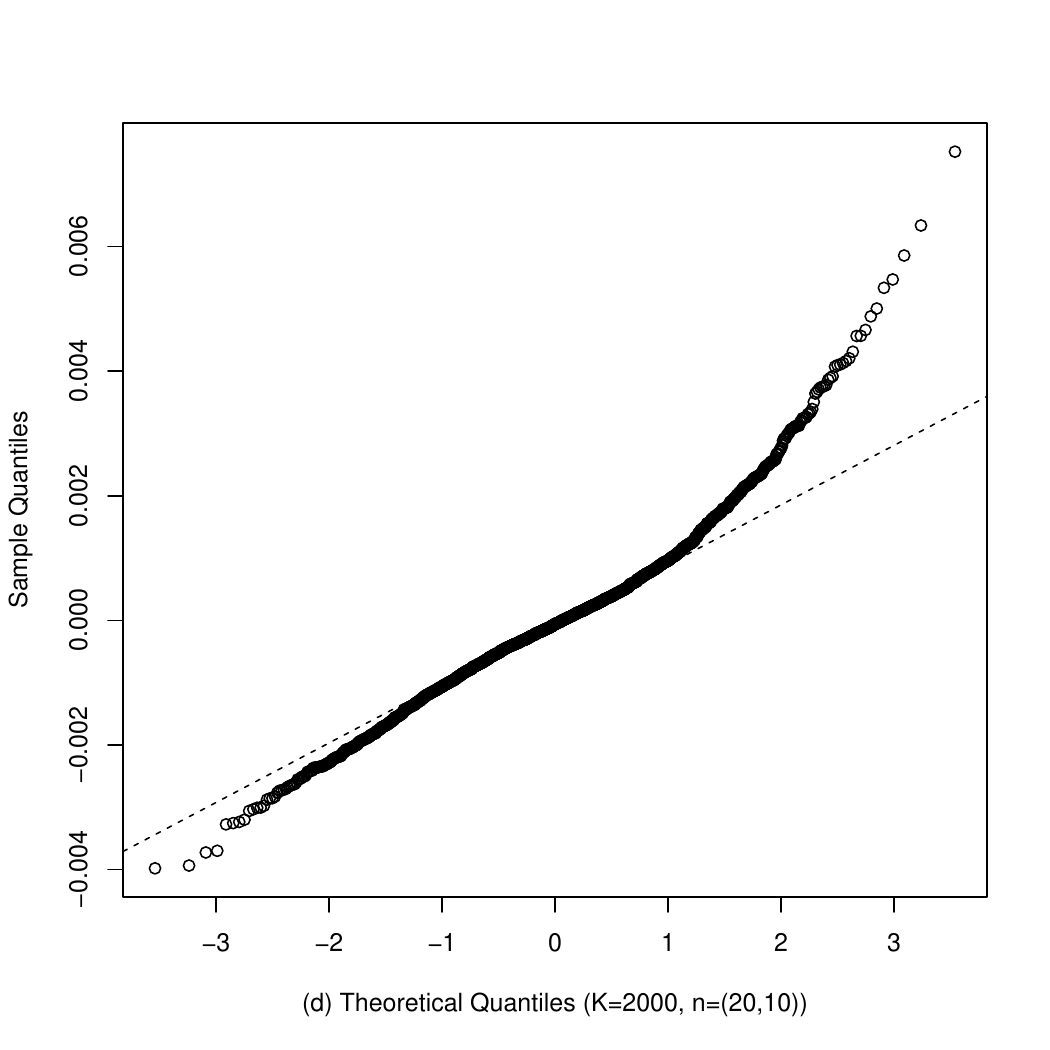}
}

\subfloat{%
\includegraphics[width=7.5cm,height=5.5cm]{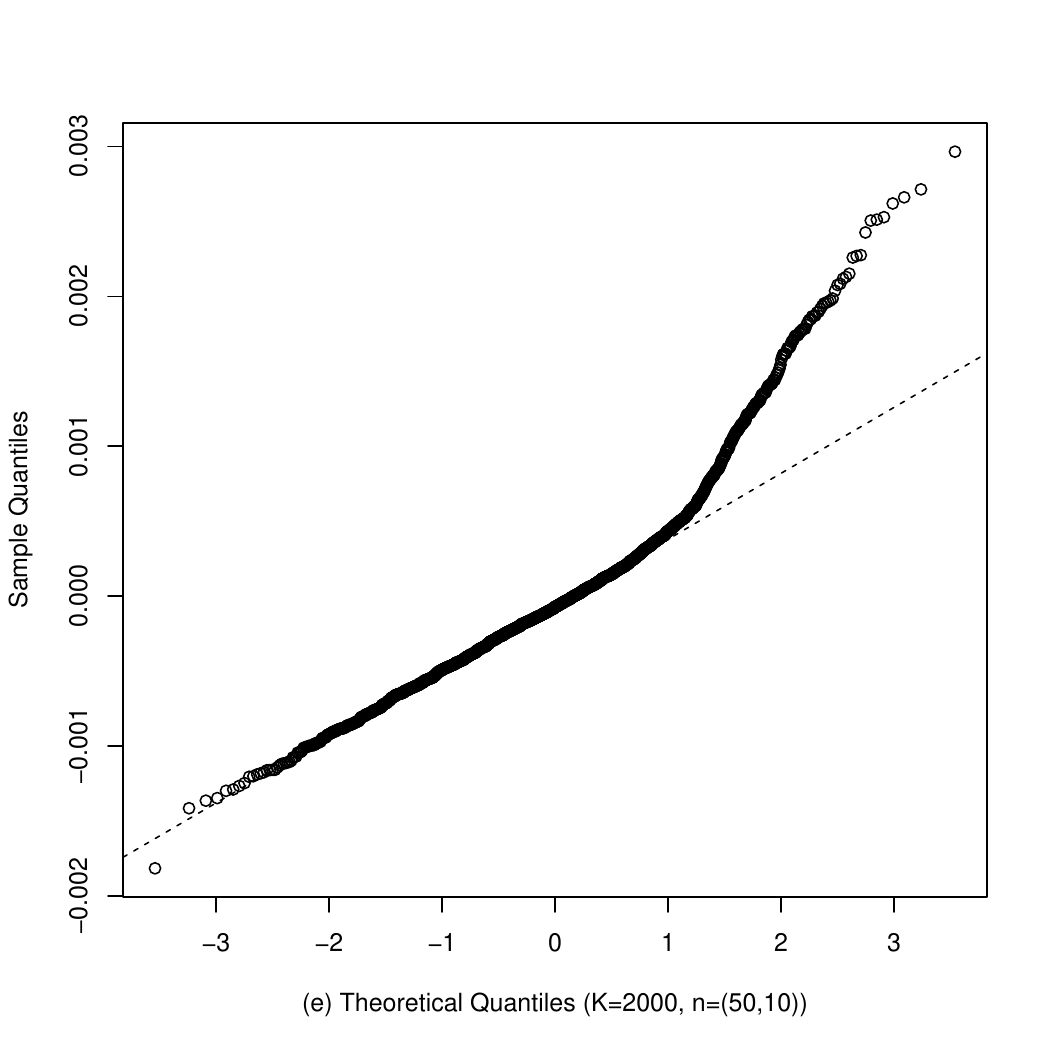}
}
\quad
\subfloat{%
\includegraphics[width=7.5cm,height=5.5cm]{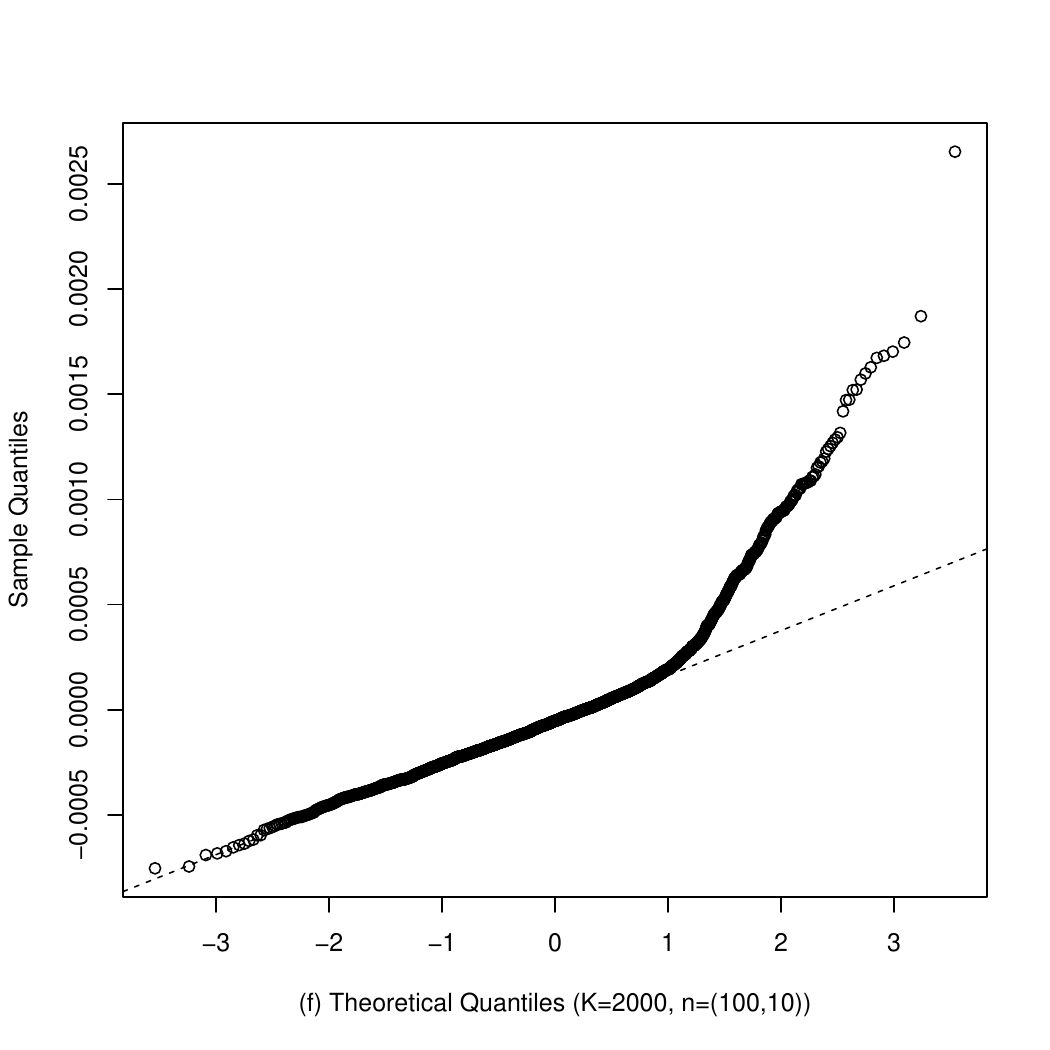}
}

\caption{Normal Q-Q plot of distribution of the test statistic under $H_0$, HKY model, $G=2$ and $K=2000$ (a) n=(20,20), (b) n=(50,50), (c) n=(100,100), (d) n=(20,10), (e) n=(50,10) and (f) n=(100,10).}
\label{fig:QQHKYG2}
\end{figure}

\FloatBarrier


\begin{figure}[htb]
\centering
\subfloat{%
\includegraphics[width=7.5cm,height=5.5cm]{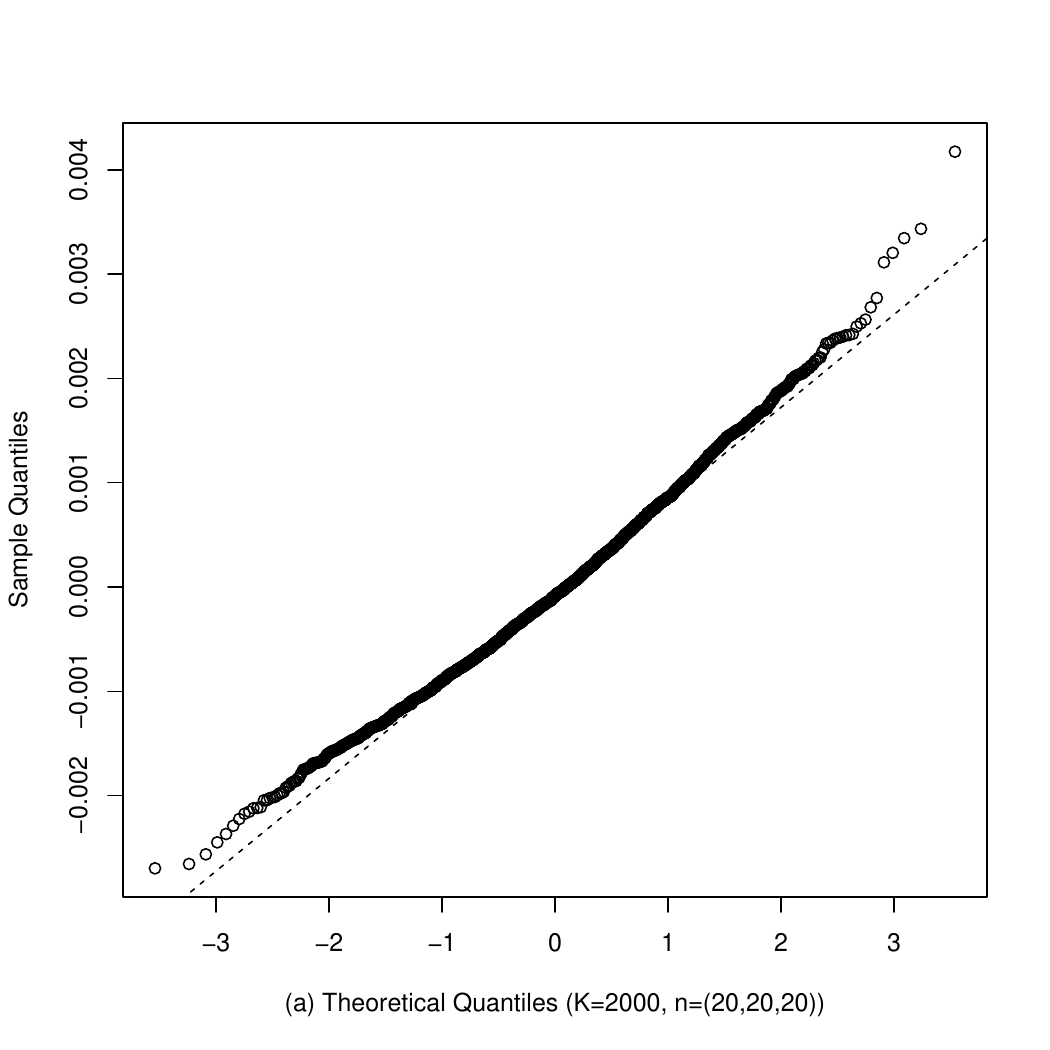}
}
\quad
\subfloat{%
\includegraphics[width=7.5cm,height=5.5cm]{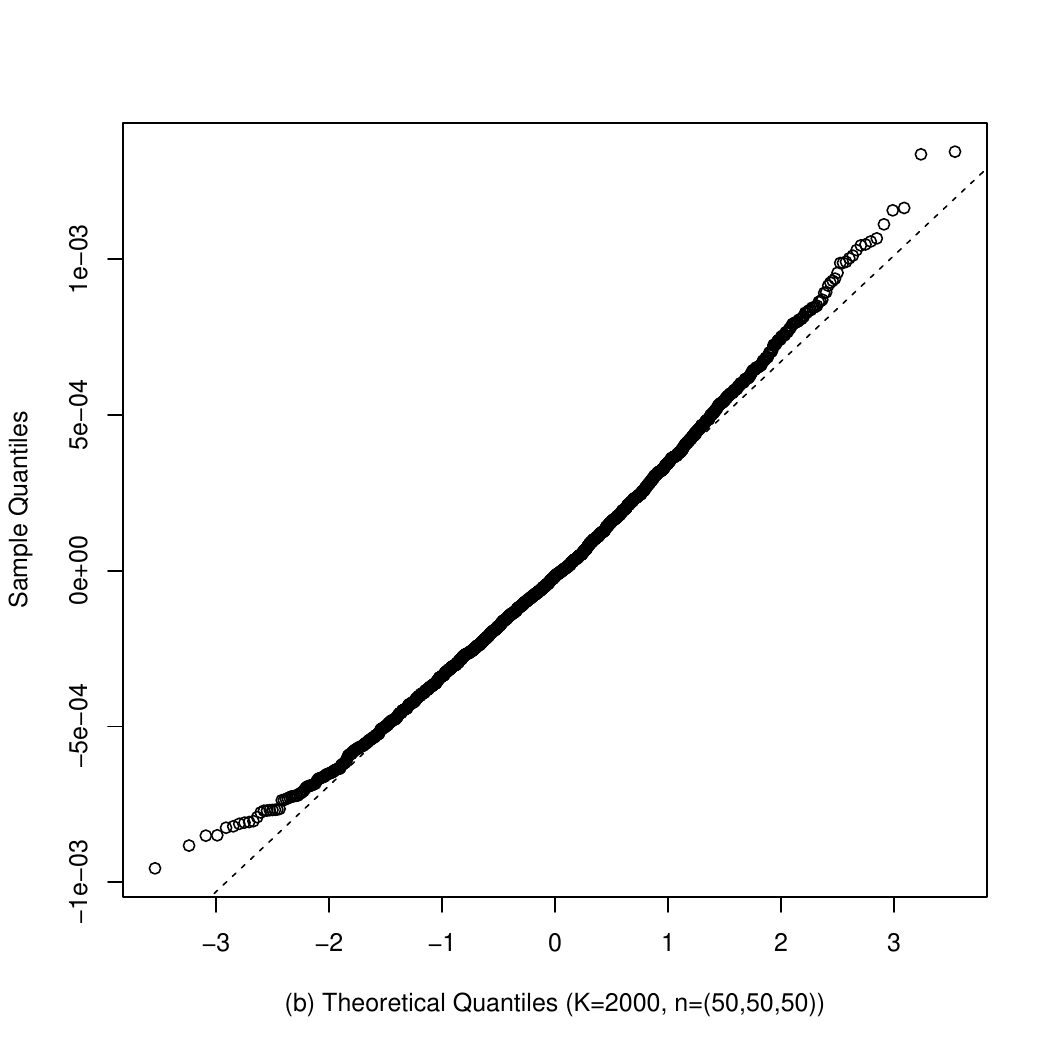}
}

\subfloat{%
\includegraphics[width=7.5cm,height=5.5cm]{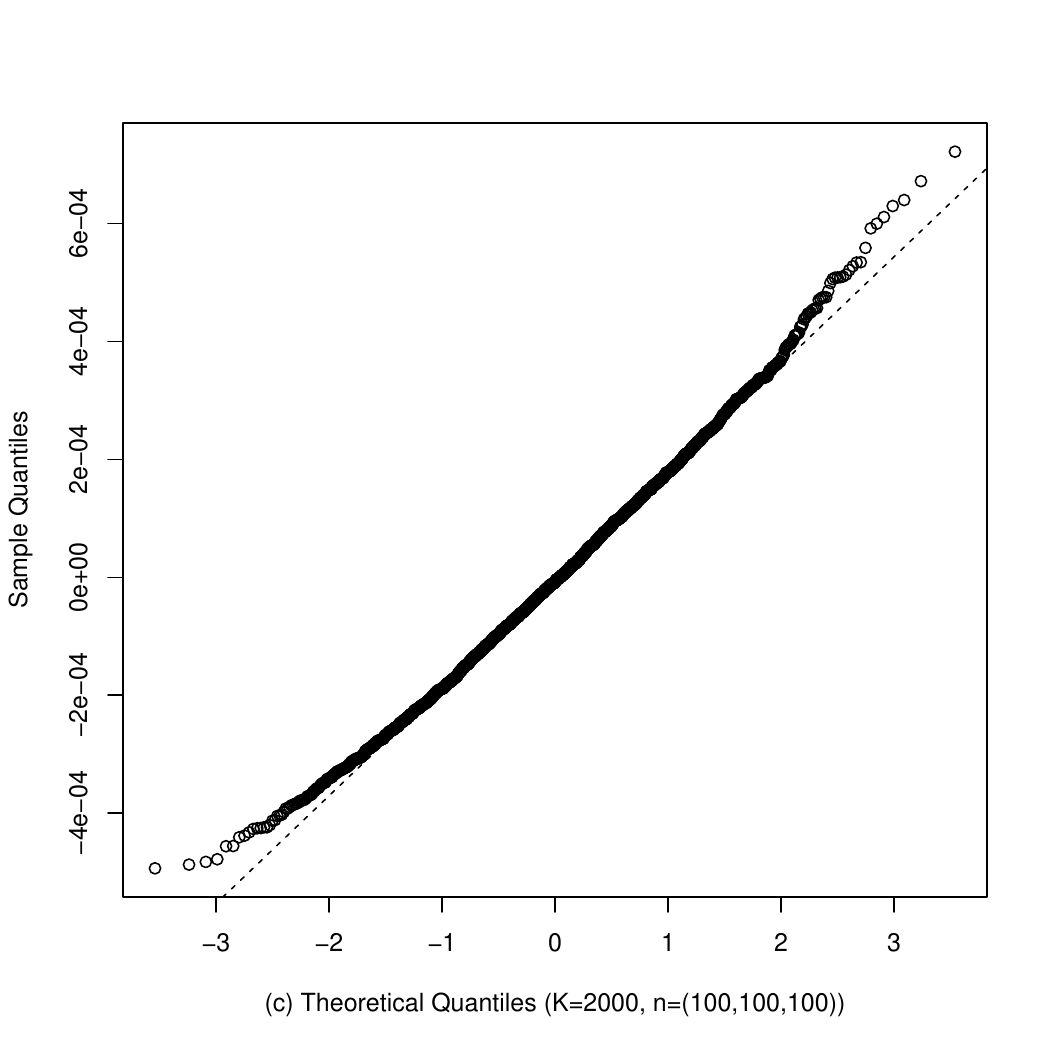}
}
\quad
\subfloat{%
\includegraphics[width=7.5cm,height=5.5cm]{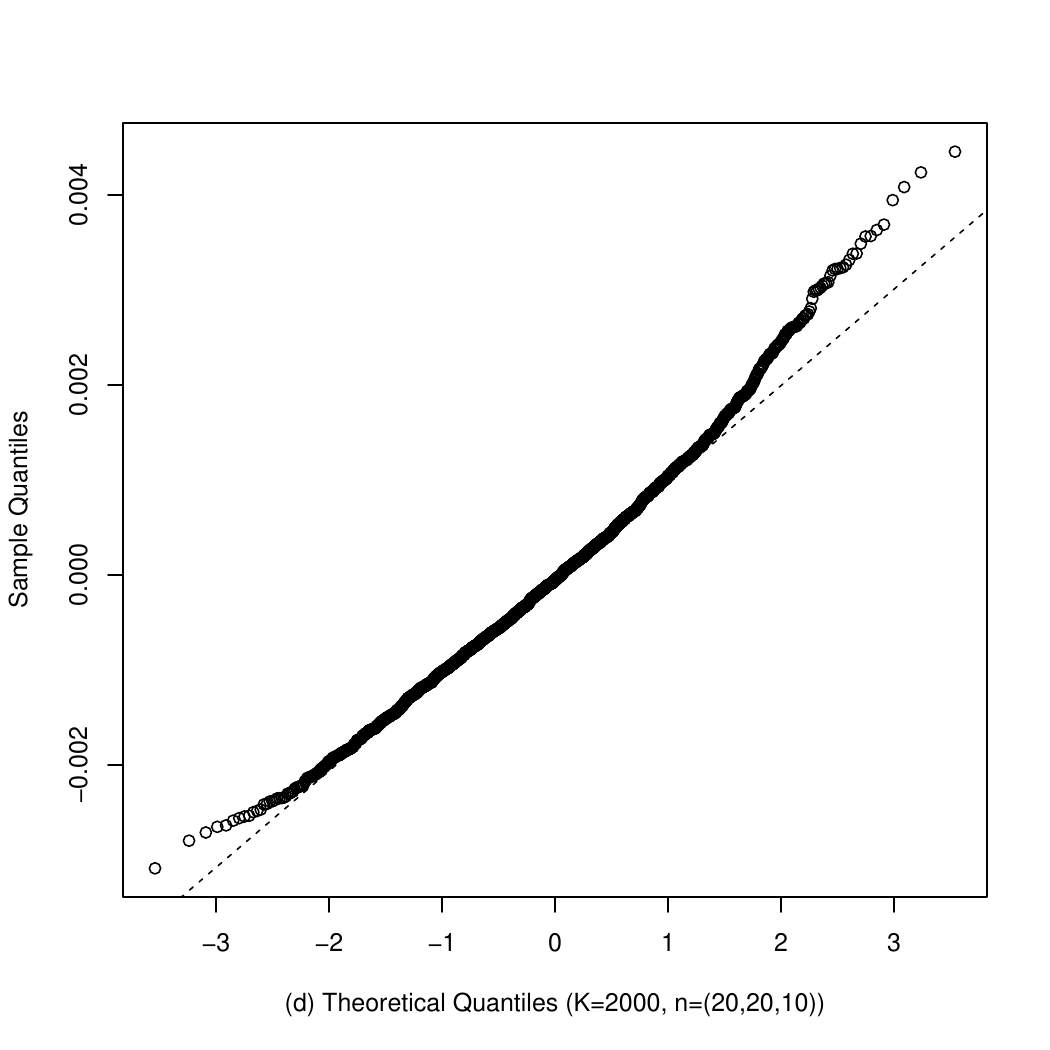}
}

\subfloat{%
\includegraphics[width=7.5cm,height=5.5cm]{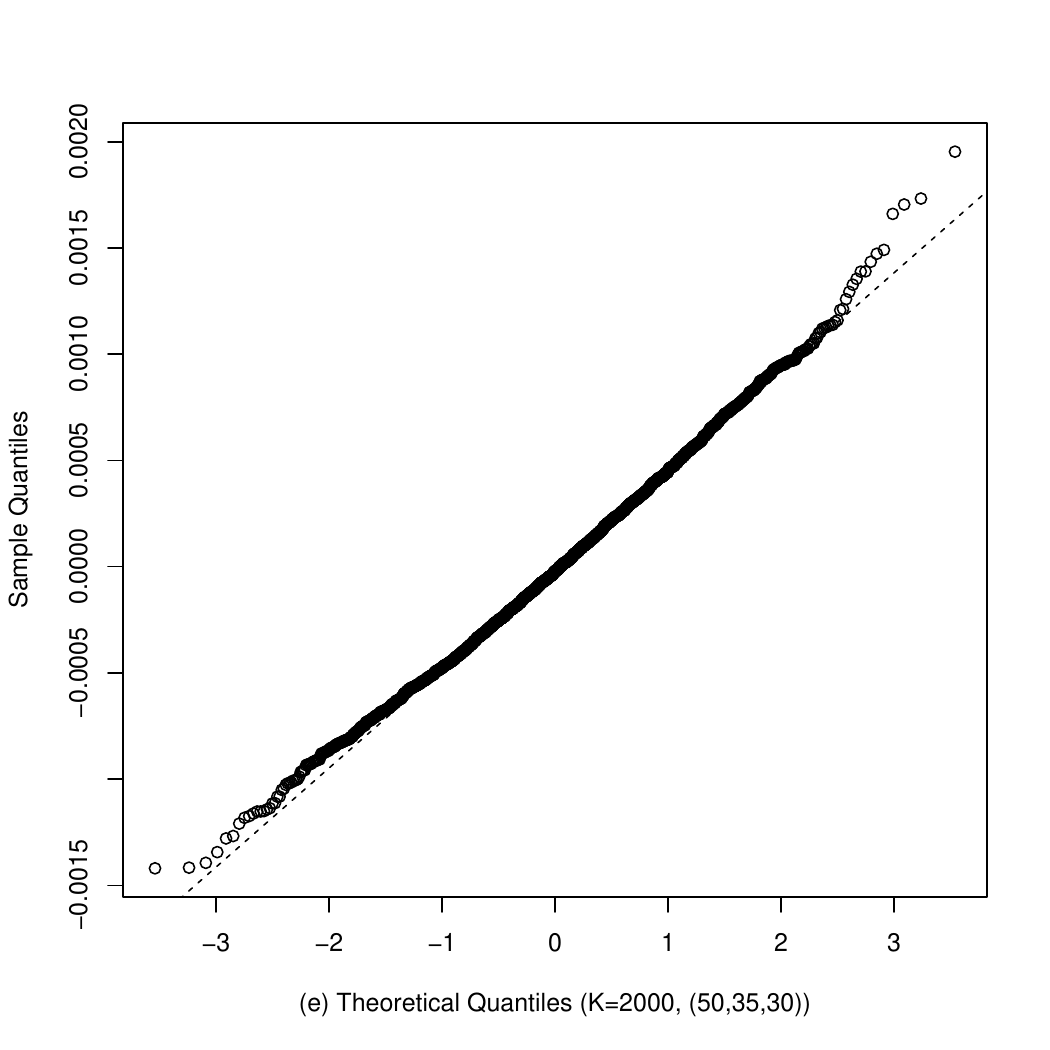}
}
\quad
\subfloat{%
\includegraphics[width=7.5cm,height=5.5cm]{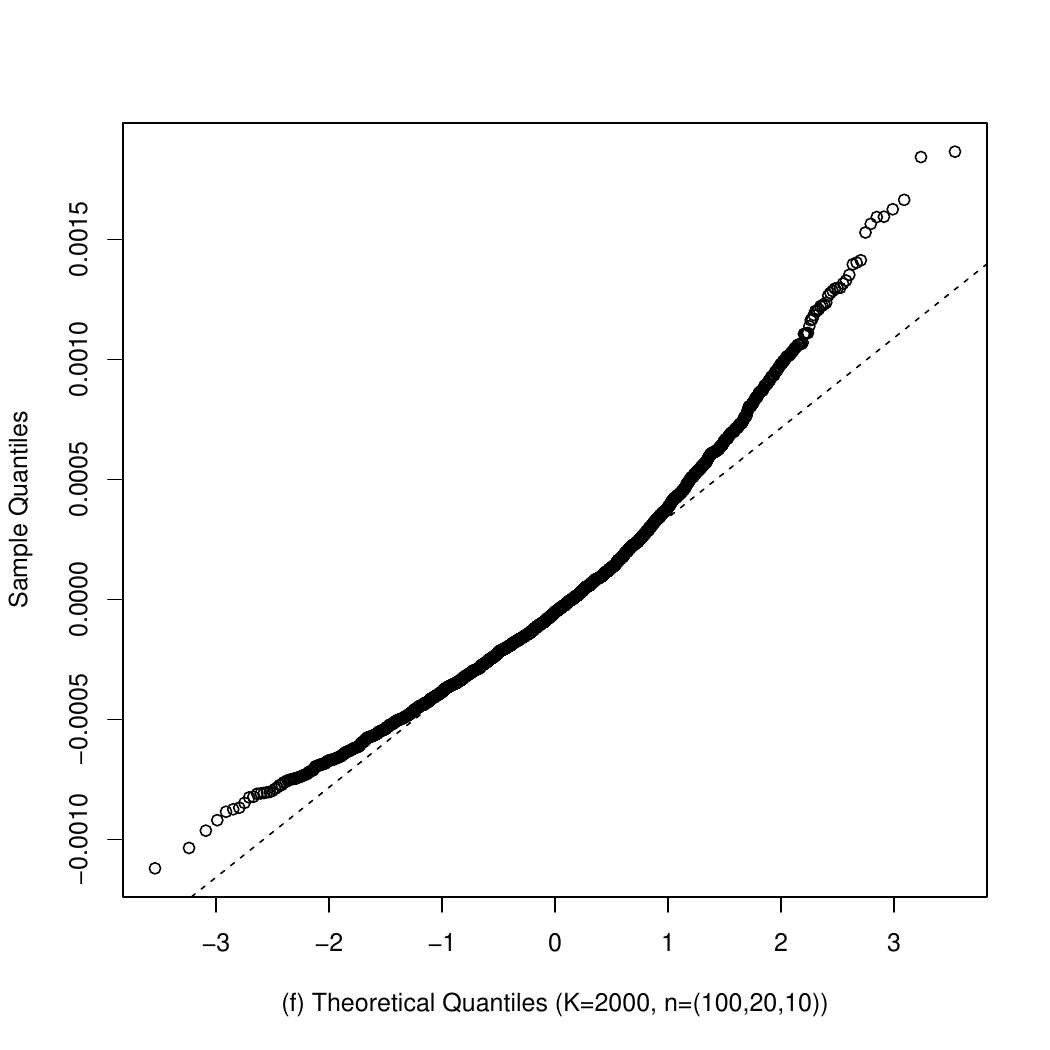}
}

\caption{Normal Q-Q plot of distribution of the test statistic under $H_0$, HKY model, $G=3$ and $K=2000$ (a) n=(20,20,20), (b) n=(50,50,50), (c) n=(100,100,100), (d) n=(20,20,10), (e) n=(50,35,30) and (f) n=(100,20,10).}
\label{fig:QQHKYG3}
\end{figure}

\FloatBarrier


\begin{figure}[htb]
\centering
\subfloat{%
\includegraphics[width=7.5cm,height=5.5cm]{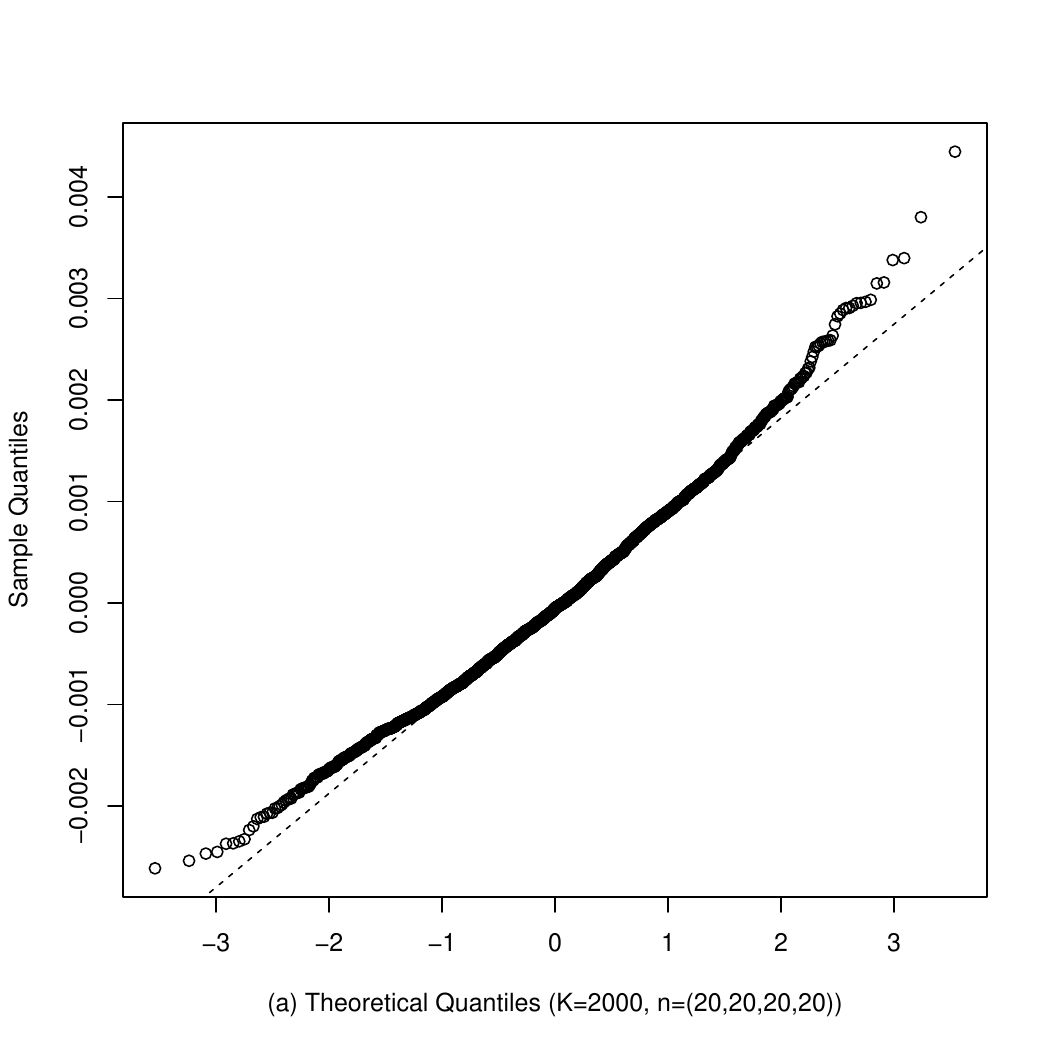}
}
\quad
\subfloat{%
\includegraphics[width=7.5cm,height=5.5cm]{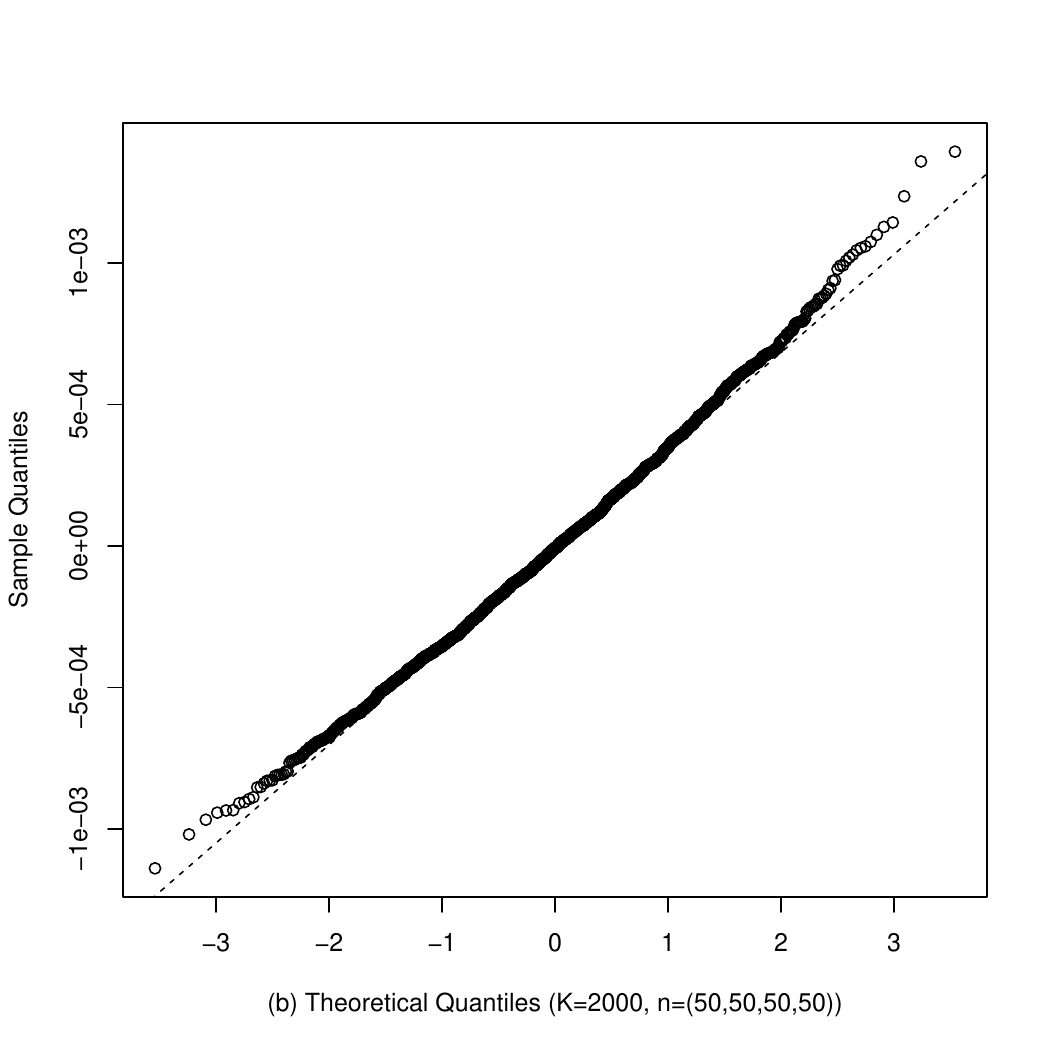}
}

\subfloat{%
\includegraphics[width=7.5cm,height=5.5cm]{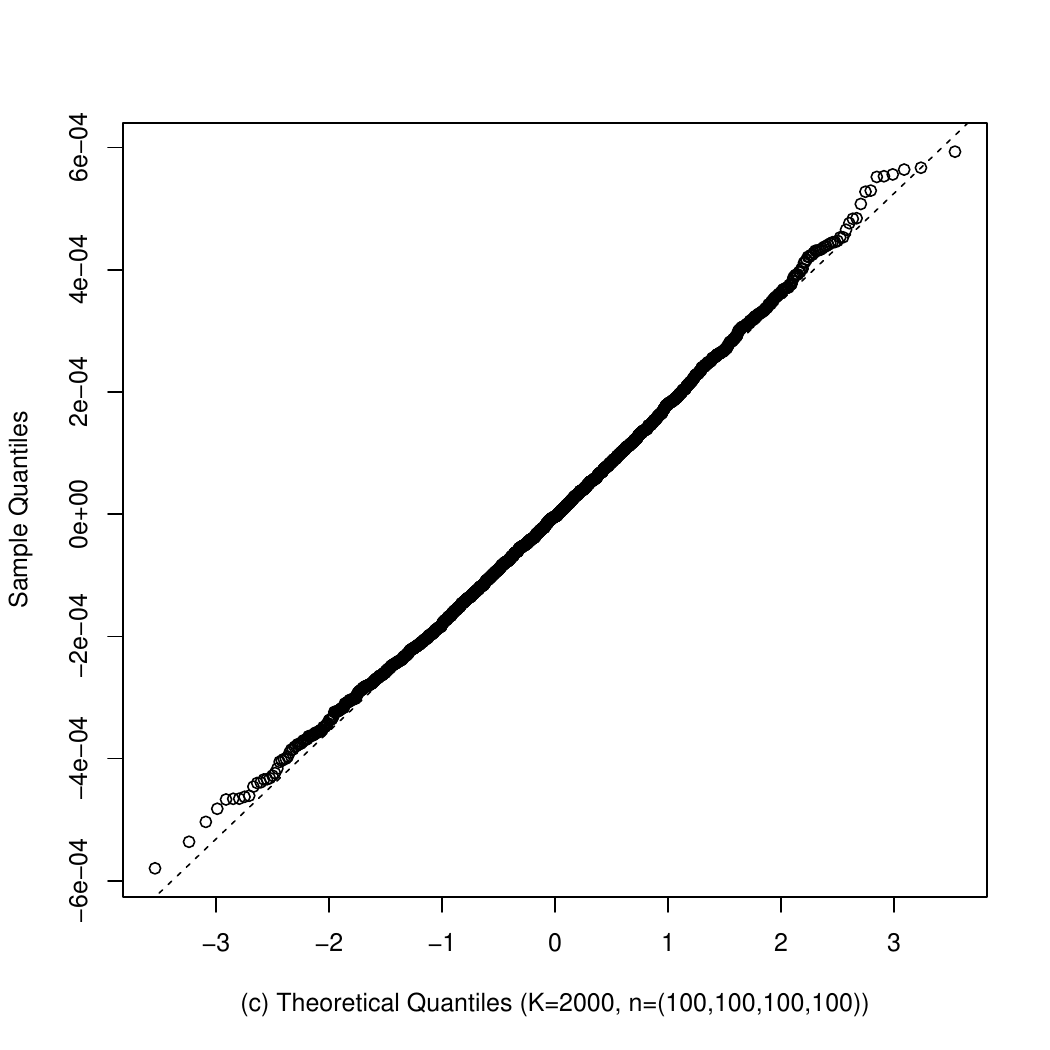}
}
\quad
\subfloat{%
\includegraphics[width=7.5cm,height=5.5cm]{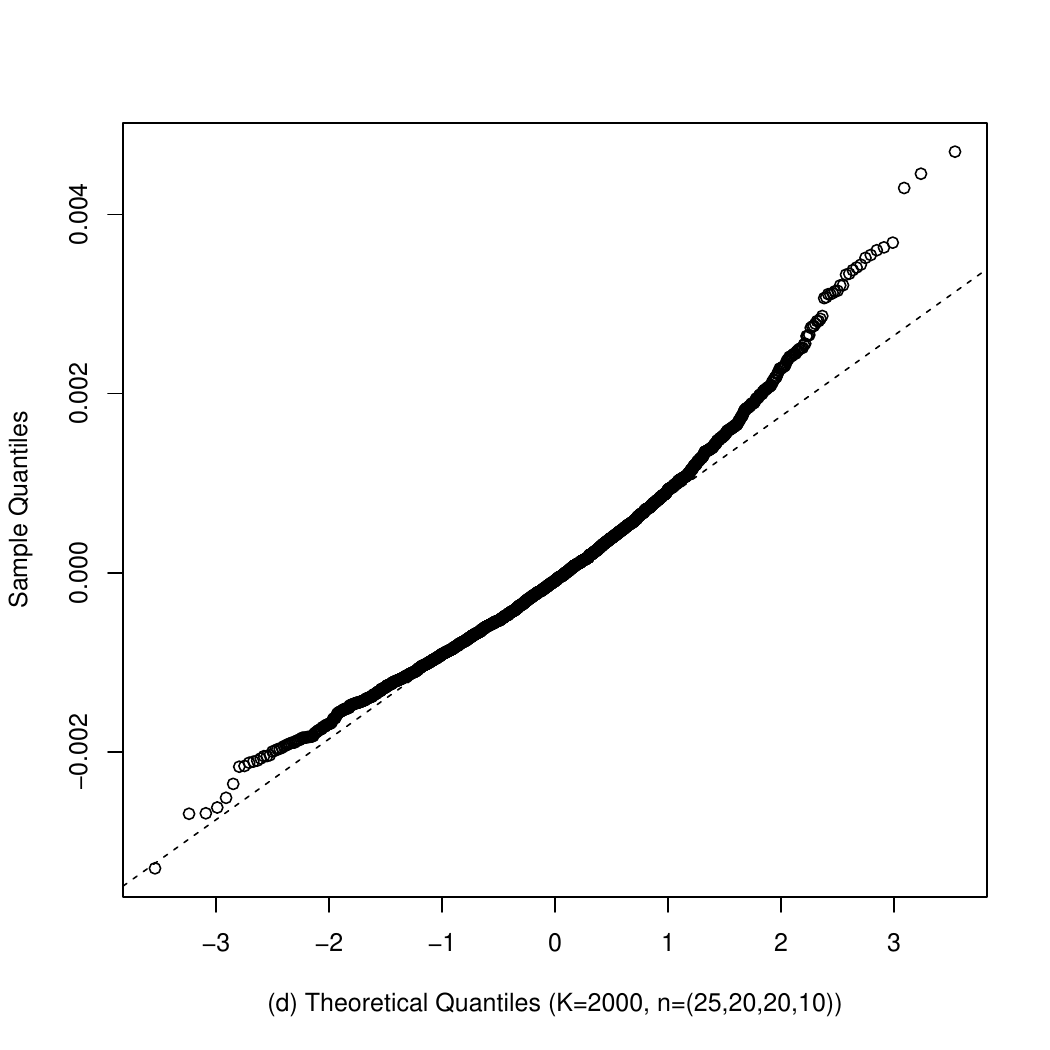}
}

\subfloat{%
\includegraphics[width=7.5cm,height=5.5cm]{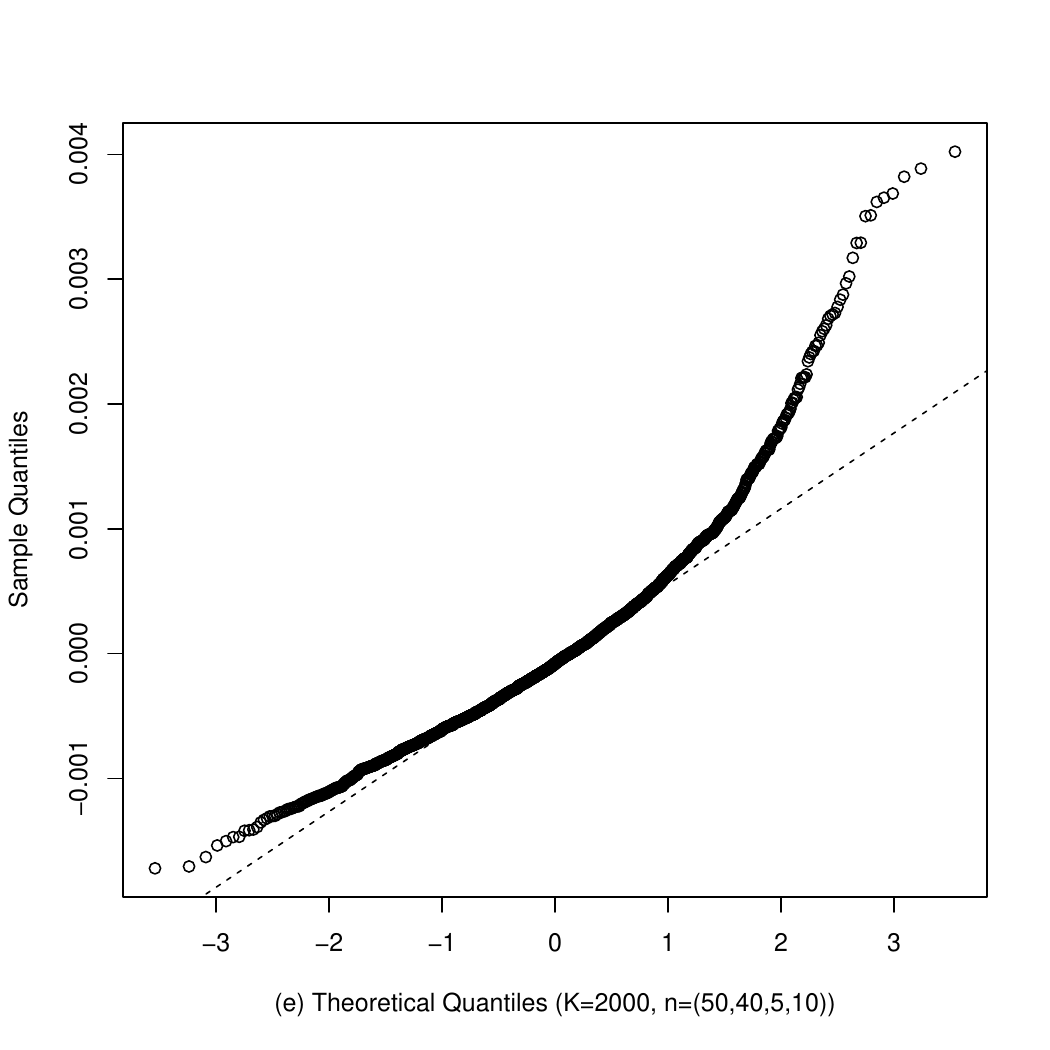}
}
\quad
\subfloat{%
\includegraphics[width=7.5cm,height=5.5cm]{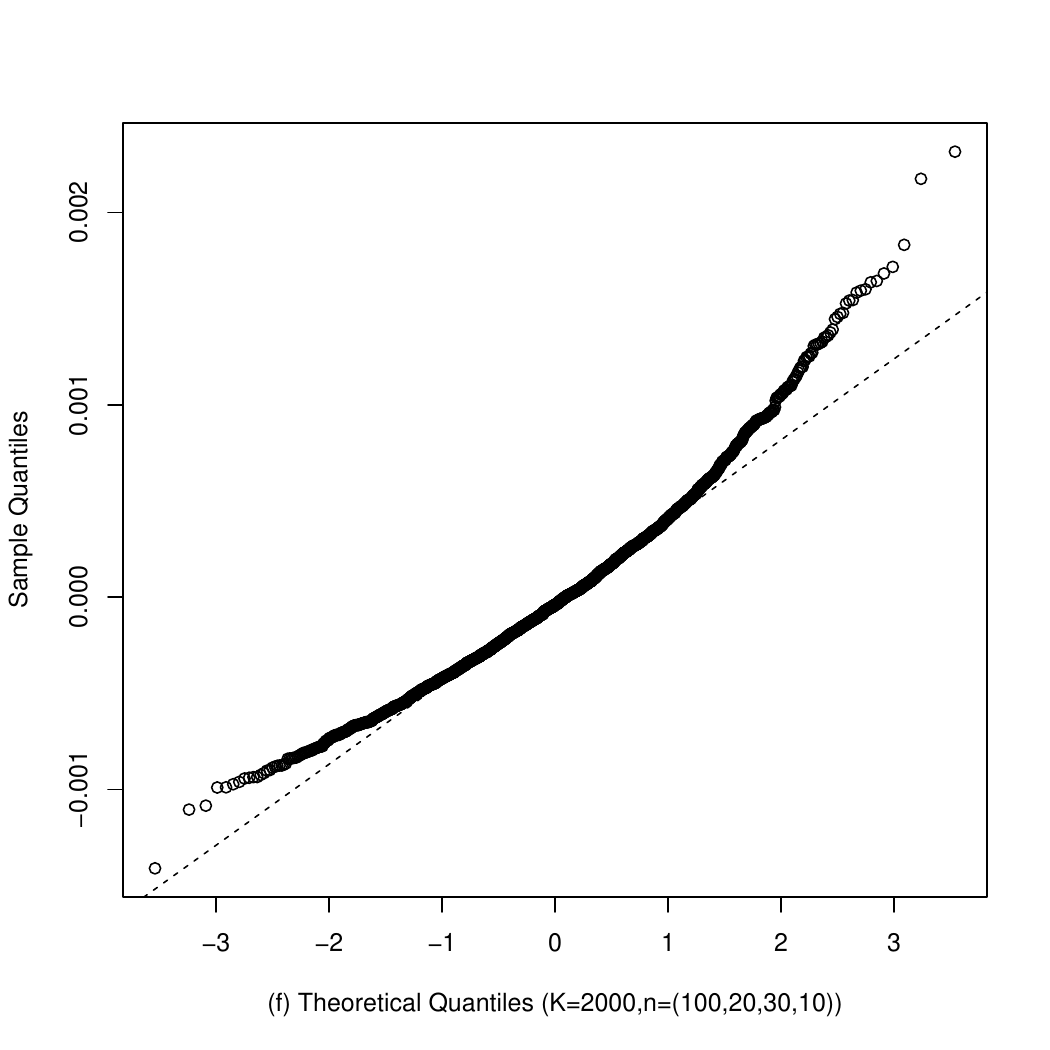}
}

\caption{Normal Q-Q plot of distribution of the test statistic under $H_0$, HKY model, $G=4$ and $K=2000$ (a) n=(20,20,20,20), (b) n=(50,50,50,50), (c) n=(100,100,100,100), (d) n=(25,20,20,10), (e) n=(50,40,5,10) and (f) n=(100,20,30,10).}
\label{fig:QQHKYG4}
\end{figure}

\FloatBarrier

Regarding the symmetry of the distribution of $D_n(t, B)$, the results improve 
as $n$ increases. 
A simulation was performed considering $K=2000$, for $G=2$, for balanced cases ($n=200, 500$ and $3000$) for each group. In Figure \ref{fig:QQ-sampleHKYG2} Normal Q-Q plots for HKY model are shown and Shapiro Wilk's (S-W) Normality tests of $D_n(t,B)$ with $p$-values are shown in Table \ref{p_valor_JCG2} for Jukes-Cantor (JC), Kimura (K), Felsenstein (F) and HKY with $K=1000$ and $2000$. 
Considering the S-W Normality test, there is no evidence against Normality of $D_n(t,B)$ under $H_0$. In Figure \ref{fig:Denisity-sampleHKYG2} it is also shown the smoothed density estimates of the test statistic for these simulations. One can see that as $n$ increases the density is more symmetric and the variance decreases. 


\begin{figure}[htb]
\centering
\subfloat{%
\includegraphics[width=7.5cm,height=5.5cm]{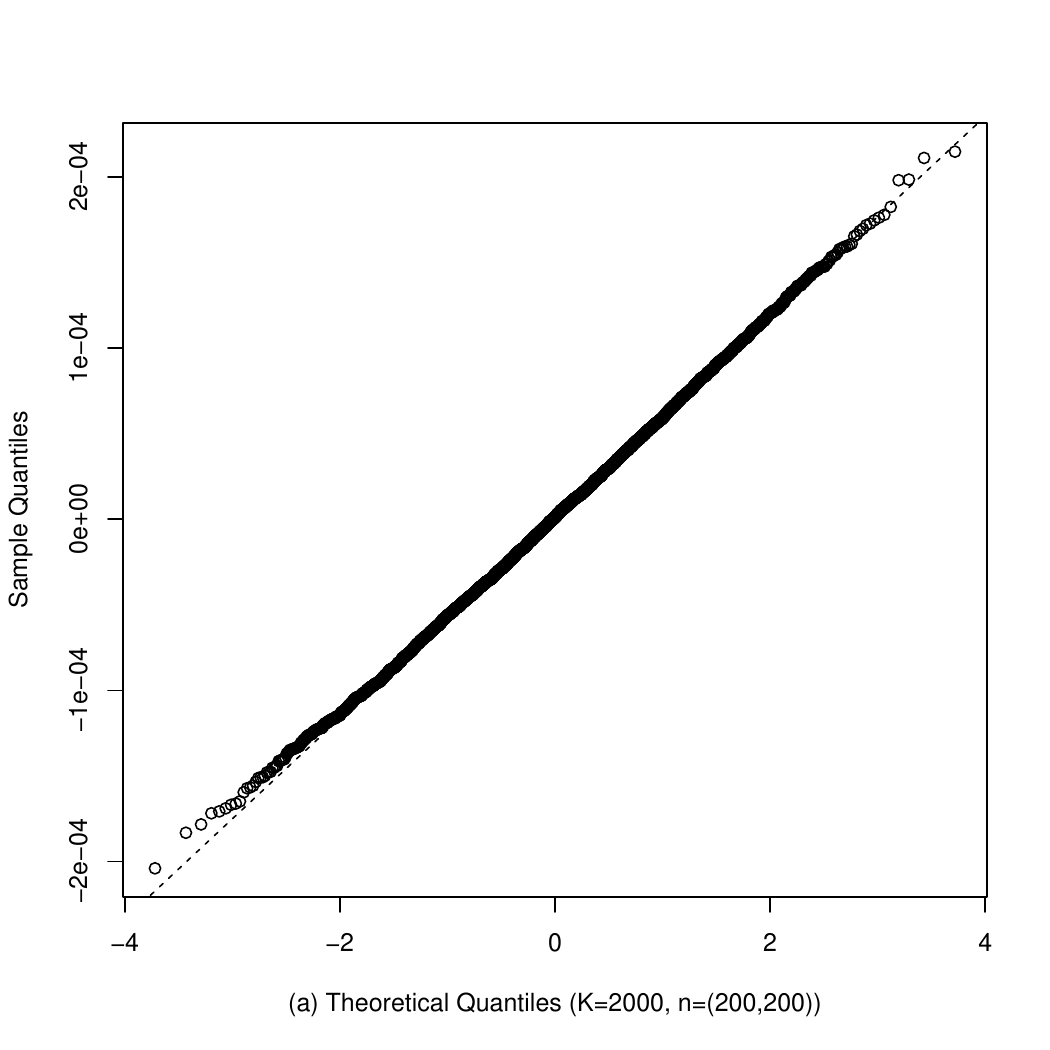}
}
\quad
\subfloat{%
\includegraphics[width=7.5cm,height=5.5cm]{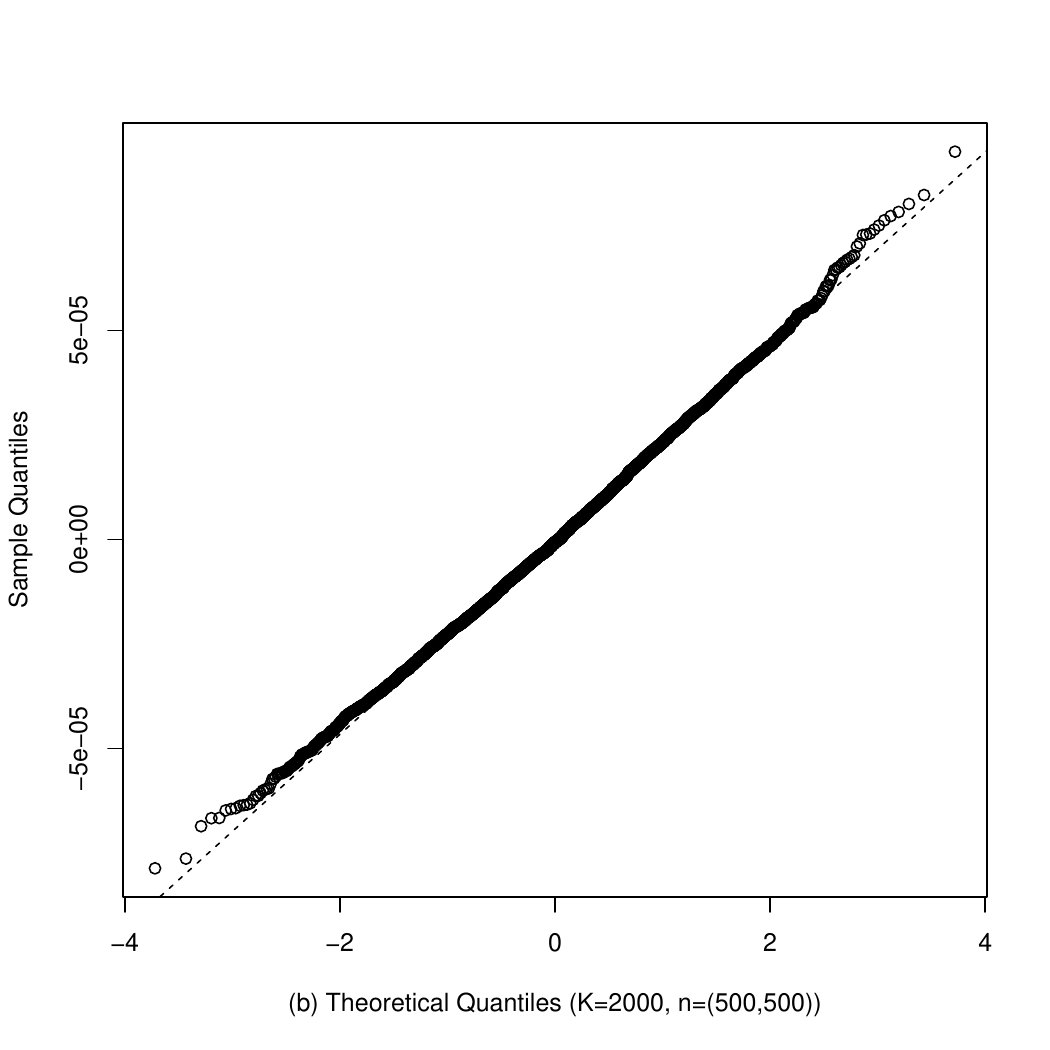}
}

\subfloat{%
\includegraphics[width=7.5cm,height=5.5cm]{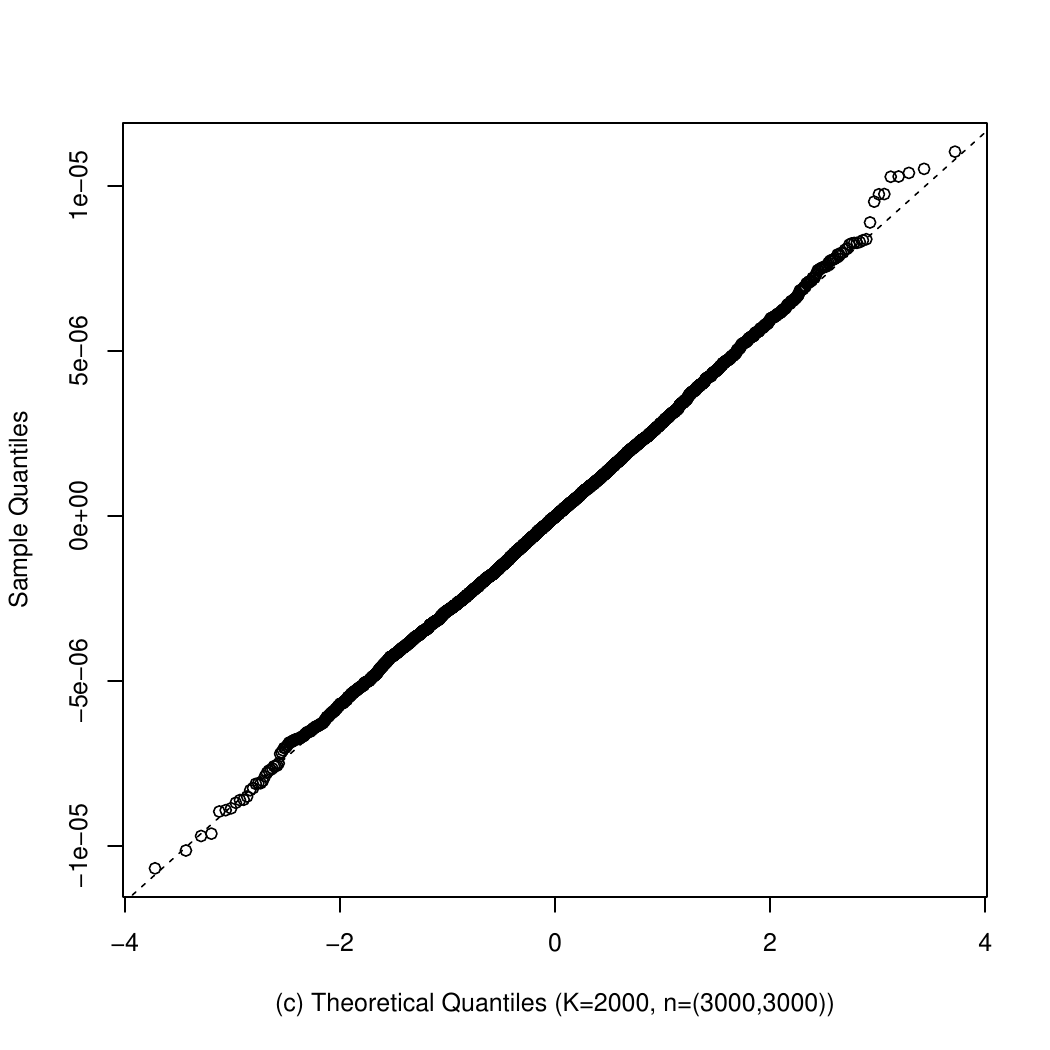}
}
\caption{Normal Q-Q plot of distribution of the test statistic under $H_0$, HKY model, $G=2$ and $K=2000$ (a) n=(200,200), (b) n=(500,500) and (c) n=(3000,3000).}
\label{fig:QQ-sampleHKYG2}
\end{figure}

\FloatBarrier

\begin{table}[htb]
\begin{center}
\caption{Shapiro-Wilk's test of normality of the test statistic under $H_0$.}
	\label{p_valor_JCG2}
  \begin{tabular}{c|c c c c c}
  \hline
	\hline
	& & \multicolumn{4}{c}{\textit{p}-value}\\
	\cline{3-6}\\
Sites          & \textbf{n} & JC & K & F & HKY\\					
	\hline
& $(200,200)$& 0.6739 & 0.6968 & 0.1642& 0.4838\\
K=1000& $(500,500)$ & 0.5463 & 0.306 & 0.0066 & 0.4793\\
                      &$(1500,1500)$&0.4207 & 0.5169 & 0.4526 &0.329\\
                      &$(3000,3000)$ &0.2064&0.1305 &  0.1001 & 0.7142 \\
                      								\hline
&$(200,200)$&0.3997 &0.3345& 0.1096 &0.7283\\
K=2000&$(500,500)$ & 0.0383 & 0.5632 & 0.0149 &0.1035\\
           &$(1500,1500)$ & 0.0405 & 0.0298 & 0.0425 &0.2197\\
                      &$(3000,3000)$&0.0252 & 0.1122 & 0.2525&0.8096\\
\hline
	\hline
  \end{tabular}
	\end{center}
\end{table}

\begin{figure}[htb]
\centering
\subfloat{%
\includegraphics[width=7.5cm,height=5.5cm]{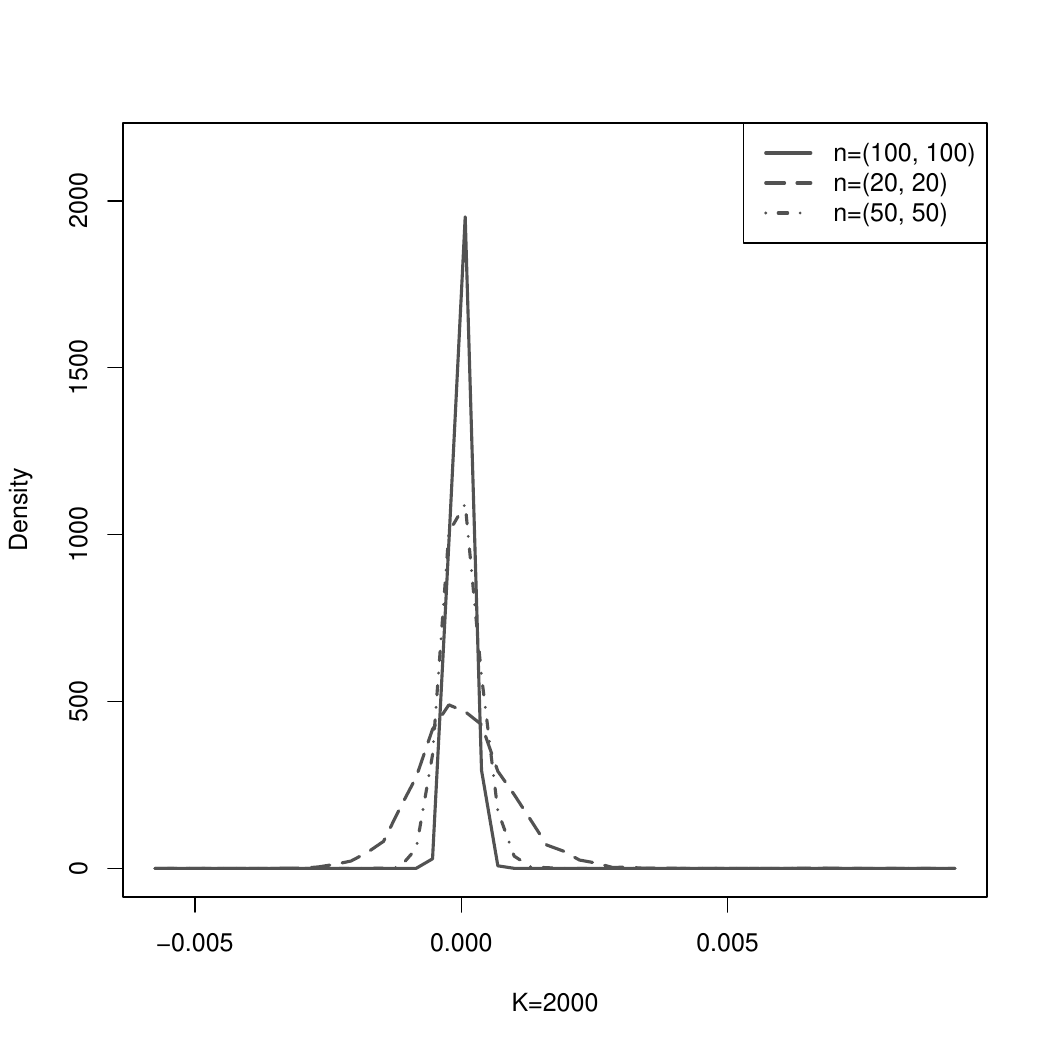}
}
\quad
\subfloat{%
\includegraphics[width=7.5cm,height=5.5cm]{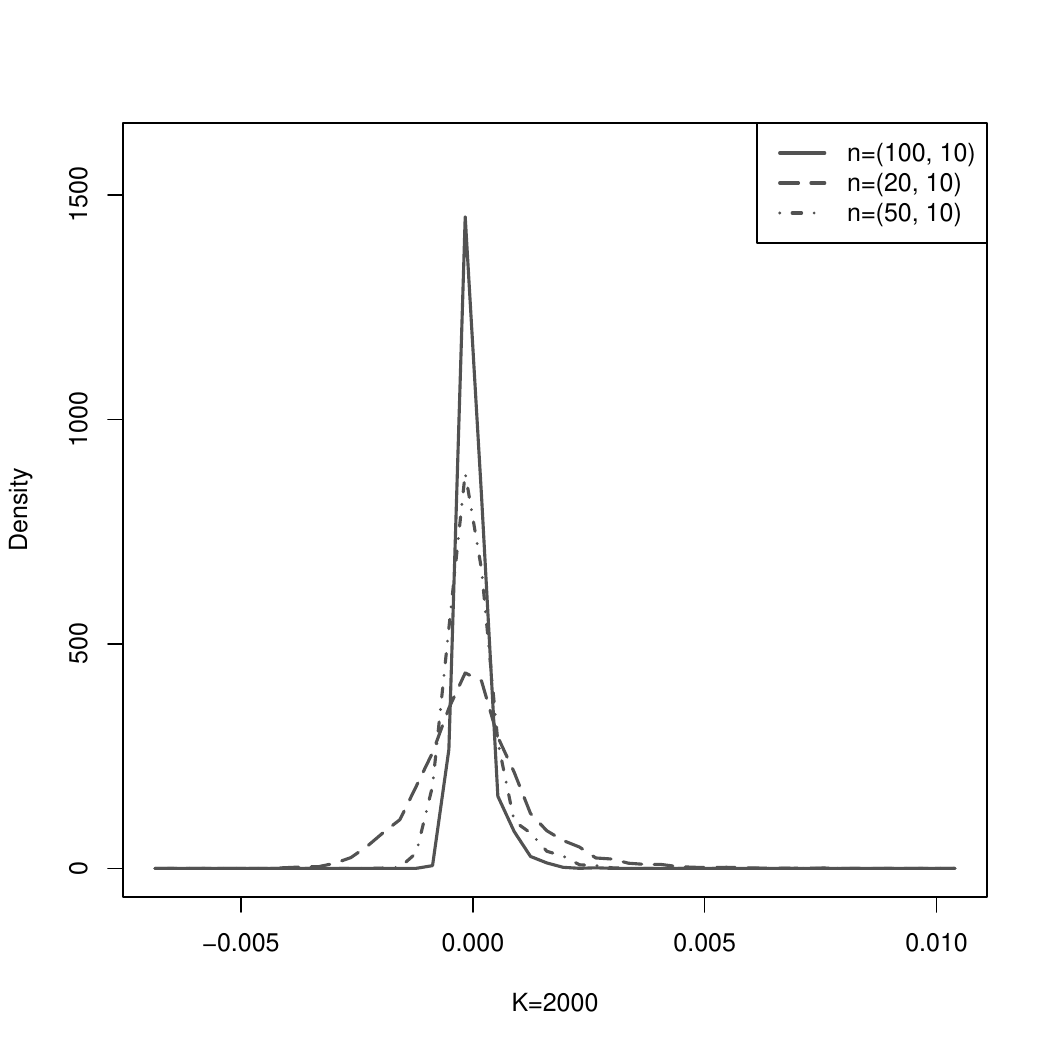}
}

\subfloat{%
\includegraphics[width=7.5cm,height=5.5cm]{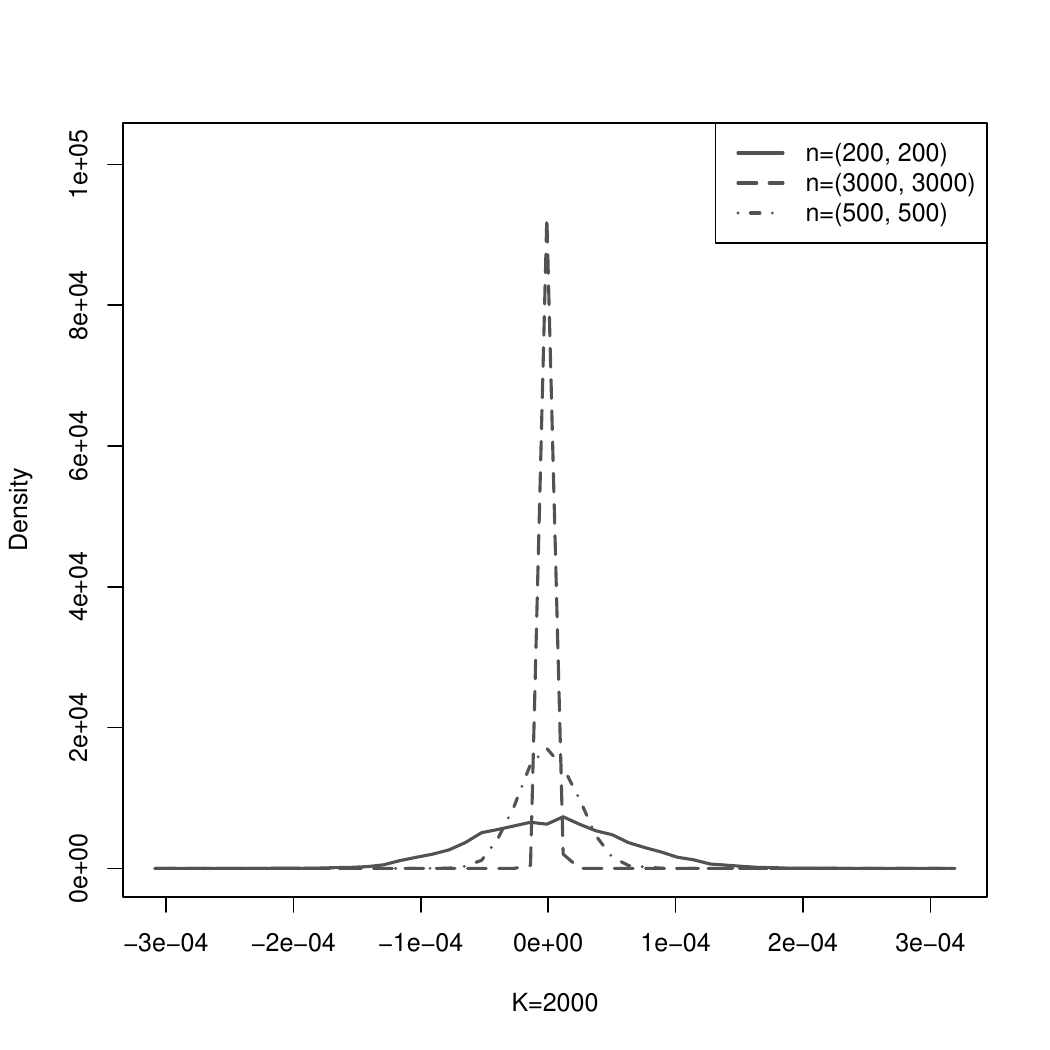}
}
\caption{Smoothed density of the test statistic under $H_0$, HKY model, $G=2$ and $K=2000$.}
\label{fig:Denisity-sampleHKYG2}
\end{figure}

\FloatBarrier

\section{Discussion}\label{sec5}

We present a test statistic for the comparison of DNA sequences under some of the most popular evolutionary processes available in the literature. Theoretical properties for the test statistic as well as its empirical performance by stochastic simulations are presented. 

The proposed test statistic is a generalized $U$-statistics built for tests of distributional homogeneity under the null hypothesis. We show that a dicothomous situation exists here. Under the null hypothesis, the $U$-statistics kernel is first-order degenerated, this test statistic falls in the quasi $U$-statistics class and follows an asymptotic normal law, albeit of higher order than the standard case. Under heterogeneity, the asymptotic normality is attained on the more usual first-order asymptotics. 

Two other issues are analyzed. First, the asymptotic normality is proven for the cases: high-dimension/large sample size,  high-dimension/small sample size, low-dimension/large sample size. The second issue regards the test power under local alternatives. The first/second order asymptotics according to the null/alternative hypotheses raises the issue of Pitman closeness and contiguity. The Pitmann class of local alternatives are discussed, and the contiguity of the test statistic for them is established.

The simulation studies help on the question of the asymptotic normality as a function of balanced/unbalanced samples, dimension and sample size. Large sample sizes (a few thousands) guarantees 
empirical asymptotic normality for any dimension. The test statistic empirical behavior may be asymmetric for smaller sample sizes if the dimension is not very high. That brings us to the cautionary measure of employing  resampling techniques to accurately calculating p-values.  

Summarizing our results, we present a test statistic which can be employed for ultra-high dimensional data such as the ones usually found in genetic studies. It is proven that homogeneity test statistics will be asymptotic normal under null/alternative hypotheses and even for local alternatives. The genetic sequences are supposed to follow the most commonly used stochastic evolutionary processes in the literature.

\end{document}